\documentclass{lmcs} 
\pdfoutput=1

\usepackage{lastpage}
\lmcsdoi{16}{1}{9}
\lmcsheading{}{\pageref{LastPage}}{}{}%
{May~09,~2018}{Feb.~10,~2020}{}

\ACMCCS{
[{\bf Theory of computation}]:~Logic---Finite Model Theory
[{\bf Theory of computation}]:~Computational complexity and cryptography---Complexity classes}

\usepackage{hyperref}
\usepackage{tikz}
\usepackage{cite}
\usepackage{amsmath,amsfonts,amssymb,mathabx}
\usepackage{stmaryrd}
\usepackage[textwidth=2cm,textsize=small]{todonotes}
\usepackage{enumitem}

\usepackage{framed}
\renewenvironment{leftbar}[1][\hsize]
{%
	\MakeFramed{\hsize#1\advance\hsize-\width\FrameRestore}%
}
{\endMakeFramed}

\usepackage{mathrsfs}

\usetikzlibrary{chains,fit,shapes}
\usetikzlibrary{arrows,positioning} 

\newcommand{\op}[1]{\operatorname{#1}}

\newcommand{\bbN}{\mathbb{N}}

\newcommand{\bbZ}{\mathbb{Z}}

\newcommand{\cB}{\mathcal{B}}

\newcommand{\cG}{\mathcal{G}}

\newcommand{\fG}{\mathfrak{G}}

\newcommand{\fM}{\mathfrak{M}}

\newcommand{\bG}{\mathbf{G}}

\newcommand{\bM}{\mathbf{M}}

\newcommand\loge[1]{\Sigma_{#1}} 
\newcommand\logu[1]{\Pi_{#1}} 
\newcommand\logex[1]{\Sigma_{#1}\textsc{[FO]}} 
\newcommand\ehorn{\Sigma_2\textsc{-Horn}} 
\newcommand\uhorn{\Pi_1\textsc{-Horn}} 
\newcommand\E[1]{\#\Sigma_{#1}} 
\newcommand\U[1]{\#\Pi_{#1}} 
\newcommand\sfo{\#\fo} 
\newcommand\seso{\text{\sc \#($\eso$)}} 
\newcommand\sh[1]{\##1} 
\newcommand\QE[1]{\eqso(\Sigma_{#1})} 
\newcommand\QU[1]{\eqso(\Pi_{#1})} 

\def\dhsat{\textsc{DisjHornSAT}}
\def\shdhsat{\textsc{\#DisjHornSAT}}
\def\cpm{\textsc{\#PerfectMatching}}
\def\chsat{\textsc{\#HornSAT}}
\def\cdnf{\textsc{\#DNF}}
\def\ctdnf{\textsc{\#3-DNF}}

\def\ctcnf{\textsc{\#3-CNF}}

\def\csat{\textsc{\#SAT}}

\def\dotminus{\mathbin{\ooalign{\hss\raise1ex\hbox{.}\hss\cr
			\mathsurround=0pt$-$}}}

\def\F{\mathcal{F}}
\def\L{\mathcal{L}}
\def\cG{\mathcal{G}}
\def\N{\mathbb{N}}

\def\a{\bar{a}}

\def\t{\bar{t}}
\def\u{\bar{u}} 
\def\v{\bar{v}} 
\def\x{\bar{x}} 
\def\y{\bar{y}} 
\def\z{\bar{z}} 

\renewcommand{\iff}{\;\leftrightarrow\;}

\newcommand{\fa}[1]{\forall{#1}.\:}
\newcommand{\ex}[1]{\exists{#1}.\:}

\newcommand{\fo}{{\rm FO}}
\newcommand{\so}{{\rm SO}}
\newcommand{\lfp}{{\rm LFP}}

\newcommand{\clfp}[1]{[{\bf lsfp} \, #1]}
\newcommand{\fqfo}{{\rm FQFO}}
\newcommand{\fqso}{{\rm FQSO}}
\newcommand{\pth}{{\bf path} \,\, }
\newcommand{\tc}{{\rm TC}}

\newcommand{\pfp}{{\rm PFP}}
\newcommand{\eso}{\exists\so}
\newcommand{\first}{\operatorname{first}}
\newcommand{\last}{\operatorname{last}}
\newcommand{\succesor}{\operatorname{succ}}

\newcommand{\R}{\mathbf{R}}

\newcommand{\A}{\mathfrak{A}}

\newcommand{\qso}{{\rm QSO}}
\newcommand{\qsoz}{\qso_{\bbZ}}
\newcommand{\optqso}{{\rm OptQSO}}
\newcommand{\maxqso}{{\rm MaxQSO}}
\newcommand{\minqso}{{\rm MinQSO}}
\newcommand{\rqfo}{{\rm RQFO}}
\newcommand{\tqfo}{{\rm TQFO}}

\newcommand{\qfo}{{\rm QFO}}

\newcommand{\eqso}{\Sigma\qso}
\newcommand{\eqsoz}{\eqso_{\bbZ}}

\newcommand{\fv}{\mathbf{FV}}
\newcommand{\sv}{\mathbf{SV}}
\newcommand{\fs}{\mathbf{FS}}
\newcommand{\arity}{{\rm arity}}
\newcommand{\length}[1]{\vert #1 \vert}
\newcommand{\size}[1]{\vert #1 \vert}
\newcommand{\shp}{\text{\sc \#P}}
\newcommand{\ptime}{\text{\sc P}}
\newcommand{\np}{\text{\sc NP}}
\newcommand{\bpp}{\text{\sc BPP}}
\newcommand{\cspp}{\text{\sc SPP}}

\newcommand{\rp}{\text{\sc RP}}
\newcommand{\pspace}{\text{\sc PSPACE}}
\newcommand{\nlog}{\text{\sc NL}}
\newcommand{\ulog}{\text{\sc UL}}

\newcommand{\shl}{\text{\sc \#L}}
\newcommand{\spp}{\text{\sc SpanP}}
\newcommand{\spanl}{\text{\sc SpanL}}
\newcommand{\gp}{\text{\sc GapP}}
\newcommand{\optp}{\text{\sc OptP}}
\newcommand{\maxp}{\text{\sc MaxP}}
\newcommand{\minp}{\text{\sc MinP}}
\newcommand{\maxpb}{\text{\sc MaxPB}}
\newcommand{\minpb}{\text{\sc MinPB}}
\newcommand{\fp}{\text{\sc FP}}
\newcommand{\totp}{\text{\sc TotP}}

\newcommand{\fpspace}{\text{\sc FPSPACE}}
\newcommand{\nfpspace}{\text{\sc FPSPACE(poly)}}

\newcommand{\CC}{\mathscr{C}}

\newcommand{\FF}{\mathscr{F}}

\newcommand{\LL}{\mathscr{L}}

\newcommand{\enc}{{\rm enc}}
\newcommand{\ostr}{\text{\sc OrdStruct}}

\newcommand{\nat}{\mathbb{N}}

\newcommand{\add}{+}

\newcommand{\mult}{\cdot}

\newcommand{\sem}[1]{{\llbracket{}{#1}\rrbracket}}
\newcommand{\pa}[1]{\Pi{#1}.\,}

\renewcommand{\sa}[1]{\Sigma{#1}.\,}

\newcommand{\maxa}[1]{\operatorname{Max}{#1}.\,}
\newcommand{\mina}[1]{\operatorname{Min}{#1}.\,}
\newcommand{\clique}{\operatorname{clique}}

\tikzset{
	defaultstyle/.style={>=stealth,semithick, auto,font=\small,
		initial text= {},
		initial distance= {3.5mm},
		accepting distance= {3.5mm}},
	accepting/.style=accepting by arrow,
	nstate/.style={circle, semithick,inner sep=1pt, minimum size=4mm}}

\tikzset{
	rect/.style={
		rectangle,
		rounded corners,
		draw=black, 
		thick,
		text centered},
	rectw/.style={
		rectangle,
		rounded corners,
		draw=white, 
		thick,
		text centered},
	sq/.style={
		rectangle,
		draw=black, 
		thick,
		text centered},
	sqw/.style={
		rectangle,
		draw=white, 
		thick,
		text centered},
	arrout/.style={
		->,
		-latex,
		thick,
	},
	arrin/.style={
		<-,
		latex-,
		thick,
		El         },
	arrd/.style={
		<->,
		>=latex,
		thick,
	},
	arrw/.style={
		thick,
	}
}

\newcommand{\tmo}{\text{\rm \#output}}
\newcommand{\tma}{\text{\rm \#accept}}
\newcommand{\tmr}{\text{\rm \#reject}}

\newcommand{\tmt}{\text{\rm \#total}}

\newcommand{\support}{\text{\rm supp}}
\newcommand{\supp}{\text{\rm supp}}
\newcommand{\val}{\text{\rm val}}

\theoremstyle{plain} 

\begin{document}

\title[Descriptive Complexity for Counting Complexity Classes]{Descriptive Complexity for Counting Complexity Classes}

\author[M.~Arenas]{Marcelo Arenas}	

\author[M.~Munoz]{Martin Mu\~noz}	
\address{Pontificia Universidad Cat\'olica de Chile \& IMFD Chile}	
\email{marenas@ing.puc.cl}  
\email{mmunos@uc.cl}  
\email{cristian.riveros@uc.cl}  

\author[C.~Riveros]{Cristian Riveros}	





\begin{abstract}
  \noindent Descriptive Complexity has been very successful in characterizing complexity classes of decision problems in terms of the properties definable in some logics. However, descriptive complexity for counting complexity classes, such as $\fp$ and $\shp$, has not been systematically studied, and it is not as developed as its decision counterpart. In this paper, we propose a framework based on Weighted Logics to address this issue. Specifically, by focusing on the natural numbers we obtain a logic called Quantitative Second Order Logics (QSO), and show how some of its fragments can be used to capture fundamental counting complexity classes such as $\fp$, $\shp$ and $\fpspace$, among others. We also use QSO to define a hierarchy inside $\shp$, identifying counting complexity classes with good closure and approximation properties, and which admit natural complete problems. Finally, we add recursion to QSO, and show how this extension naturally captures lower counting complexity classes such as $\shl$.
\end{abstract}

\maketitle

\section{Introduction}

%
%
%
%
%
%
%
%

The goal of descriptive complexity is to measure the complexity of a problem in terms of the logical constructors needed to express it~\cite{immerman1999descriptive}. 
The starting point of this branch of complexity theory is Fagin's theorem~\cite{F75}, which states that $\np$ is equal to existential second-order logic. Since then, many more complexity classes have been characterized in terms of logics (see \cite{G07} for a survey) and descriptive complexity has found a variety of applications in different areas~\cite{immerman1999descriptive, L04}.
For instance, Fagin's theorem was the key ingredient to define the class {\sc MaxSNP}~\cite{PY91}, which was later shown to be a fundamental class in the study of hardness of approximation \cite{ALMSS98}. 
It is important to mention here that the definition of {\sc MaxSNP} would not have been possible without the machine-independent point of view of descriptive complexity, as pointed out in~\cite{PY91}.

Counting problems differ from decision problems in that the value of a function has to be computed.
More generally, a counting problem corresponds to computing a function $f$ from a set of instances (e.g.\ graphs, formulae, etc) to natural numbers.\footnote{This value is usually associated to counting the number of solutions
	in a search problem, but here we consider a more general definition.} 
The study of counting problems has given rise to a rich theory of counting complexity classes \cite{HV95,F97,arora2009computational}. Some of these classes are natural counterparts of some classes of decision problems; for example, $\fp$ 
is the class of all functions that can be computed in polynomial time, 
the natural counterpart of $\ptime$.
However, other function complexity classes have emerged from the need to prove that some functions are difficult to compute, even though their decision counterparts can be solved efficiently. This is the case of the class $\shp$, a counting complexity class introduced in \cite{Valiant79} to prove that natural problems like counting the number of satisfying assignments of a propositional formula in DNF or the number of perfect matchings of a bipartite graph~\cite{Valiant79} are difficult, namely, $\shp$-complete, even though their decision counterparts can be solved in polynomial time.
Starting from $\shp$,
many more natural 
counting complexity classes have been defined, such as 
$\shl$, $\spp$ and $\gp$~\cite{HV95,F97}.

Although counting problems play a prominent role in computational complexity, descriptive complexity for this type of problems has not been systematically studied and it is not as developed as for the case of decision problems. Insightful characterizations of $\shp$ and some of its extensions have been provided \cite{SalujaST95,ComptonG96}. However, these characterizations do not define function problems in terms of a logic, but instead in terms of some counting problems associated to a logic like $\fo$. Thus, it is not clear how these characterizations can be used to provide a general descriptive complexity framework for counting complexity classes like $\fp$ and $\fpspace$ (the class of functions computable in polynomial space). 

In this paper, we propose to study the descriptive complexity of counting complexity classes in terms of Weighted Logics (WL)~\cite{DrosteG07}, a general logical framework that combines Boolean formulae (e.g.\ in $\fo$ or $\so$) with operations over a fixed semiring (e.g.\ $\bbN$). 
Specifically, we propose to restrict WL to the natural numbers as a fixed semiring, calling this restriction Quantitative Second Order Logics (QSO), and study its expressive power for defining counting complexity classes over ordered structures. 
As a proof of concept, we show that natural syntactical fragments of $\qso$ capture counting complexity classes like $\shp$, $\spp$, $\fp$ and $\fpspace$.
Furthermore, by slightly extending the framework we can prove that $\qso$ can also capture classes like $\gp$ and $\optp$, showing the robustness of our approach.

The next step is to use the machine-independent point of view of $\qso$ to search for subclasses of $\shp$ with some fundamental properties.
The question here is, what properties are desirable for a subclass of $\shp$?
First, it is desirable to have a class of counting problems whose associated decision versions are tractable, in the sense that one can decide in 
polynomial time whether the value of the function is greater than $0$. 
In fact, this requirement is crucial in order to find
efficient approximation algorithms for a given function (see Section~\ref{sec:syntactic}).
Second, we expect that the class is closed under basic arithmetical operations like sum, multiplication and subtraction by one. 
This is a common topic for counting complexity classes; for example, it is known that $\shp$ is not closed under subtraction by one (under some complexity-theoretical assumption). 
Finally, we want a class with natural complete problems, which characterize all problems in it.

In this paper, we give the first results towards defining subclasses of $\shp$ that are robust in terms of existence of efficient approximations, having good closure properties, and existence of natural complete problems. 
Specifically, we introduce a syntactic hierarchy inside $\shp$, called $\eqso(\fo)$-hierarchy, and we show that it is closely related to the $\fo$-hierarchy introduced in~\cite{SalujaST95}. 
Looking inside the $\eqso(\fo)$-hierarchy, we propose the class $\eqso(\logex{1})$ and show that every function in it has a tractable associated decision version, and it is closed under sum, multiplication, and subtraction by one.
Unfortunately, it is not clear whether this class admits 
a natural complete problem.
Thus, 
we also introduce a Horn-style syntactic class, inspired by~\cite{G92},
that has tractable associated decision versions and a natural complete problem.

After studying the 
structure of $\shp$, we move beyond $\qso$ by introducing new quantifiers. 
By adding variables for functions on top of $\qso$, we introduce a quantitative least fixed point operator to the logic. 
Adding finite recursion to a numerical setting is subtle since functions over natural numbers can easily diverge without finding any fixed point. 
By using the support of the functions, we give a natural halting condition that generalizes the least fixed point operator of Boolean logics.
Then, with a quantitative recursion at hand we show how to capture $\fp$ from a different perspective and, moreover, how to restrict recursion to capture lower complexity classes 
such as~$\shl$, the counting version of $\nlog$.

It is important to mention that this paper is an extension of the conference article \cite{AMR17}. In this version, we have included the complete proofs of all the results in the paper, paying special attention in showing the main techniques used to establish them. Besides, we have simplified some of the terminology used in \cite{AMR17}, with the goal of presenting the main notions studied in the paper in a simple way. 

\smallskip
\smallskip

\noindent{\bf Organization.} The main terminology used in the paper is given in Section~\ref{sec:preliminaries}. Then the logical framework is introduced in Section~\ref{sec:logic}, and it is used to capture standard counting complexity classes in Section~\ref{sec:complexity}. The structure of $\shp$ is studied in Section~\ref{sec:syntactic}. Section~\ref{sec:beyond} is devoted to define recursion in $\qso$, and to show how to capture classes below $\fp$. 
Finally, we give some concluding remarks in Section~\ref{sec:conclusions}.


\section{Preliminaries} \label{sec:preliminaries}


\subsection{Second-order logic, LFP and PFP}
A relational signature $\R$ (or just signature) is a finite set $\{R_1, \ldots, R_k\}$, where each $R_i$ ($1 \leq i \leq k$) is a relation name with an associated arity greater than 0, which is denoted by $\arity(R_i)$. A finite structure over $\R$ (or just finite $\R$-structure) is a tuple $\A = \langle A, R_1^\A, \ldots, R_k^\A \rangle$ such that $A$ is a finite set and $R_i^\A \subseteq A^{\arity(R_i)}$ for every $i \in \{1, \ldots, k\}$. Further, an $\R$-structure $\A$ is said to be ordered if $<$ is a binary predicate name in $\R$ and $<^\A$ is a linear order on $A$.
We denote by $\ostr[\R]$ the class of all finite ordered $\R$-structures. 
In this paper we only consider finite ordered structures, so we will usually omit the words finite and ordered when referring to them.

From now on, assume given disjoint infinite sets $\fv$ and $\sv$ of first-order variables and second-order variables, respectively. Notice that every variable in $\sv$ has an associated arity, which is denoted by $\arity(X)$. Then given a  signature $\R$, the set of second-order logic formulae ($\so$-formulae) over $\R$ is given by the following grammar:
\begin{align*}\ 
	\varphi \ &:= \ x = y \ \mid \ R(\bar u) \ \mid \ \top  \ \mid\  
	X(\bar v)  \ \mid
	\neg \varphi \ \mid\ 
	(\varphi \vee \varphi) \ \mid\ 
	\ex{x} \varphi \ \mid\ 
	\ex{X} \varphi
 \end{align*}
where $x,y \in \fv$, $R \in \R$, $\bar u$ is a tuple of (not necessarily distinct) variables from $\fv$ whose length is $\arity(R)$, $\top$ is a reserved symbol to represent a tautology, $X \in \sv$, $\bar v$ is a tuple of (not necessarily distinct) variables from $\fv$ whose length is $\arity(X)$, and $x \in \fv$. 



We assume that the reader is familiar with the semantics of $\so$, so we only introduce here some notation that will be used in this paper. 
Given a signature $\R$ and an $\R$-structure $\A$ with domain $A$, a first-order assignment $v$ for $\A$ is a total function from $\fv$ to $A$, while a second-order assignment $V$ for $\A$ is a total function with domain $\sv$ that maps each $X \in \sv$ to a subset of $A^{\arity(X)}$. Moreover, given a first-order assignment $v$ for $\A$, $x \in \fv$ and $a \in A$, we denote by $v[a/x]$ a first-order assignment such that $v[a/x](x) = a$ and $v[a/x](y) = v(y)$ for every $y \in \fv$ distinct from $x$. Similarly, given a second-order assignment $V$ for $\A$, $X \in \sv$ and $B  \subseteq A^{\arity(X)}$, we denote by $V[B/X]$ a second-order assignment such that $V[B/X](X) = B$ and $V[B/X](Y) = V(Y)$ for every $Y \in \sv$ distinct from $X$. We use notation $(\A, v, V) \models \varphi$ to indicate that structure $\A$ satisfies $\varphi$ under $v$ and $V$.

In this paper, we consider several fragments or extensions of $\so$ like first-order logic~($\fo$), least fixed point logic (LFP) and partial fixed point logic (PFP) \cite{L04}. Moreover, for every $i \in \N$, we consider the fragment $\Sigma_i$ (resp., $\Pi_i$) of $\fo$, which is the set of $\fo$-formulae of the form 
$\exists \bar x_1 \forall \bar x_2 \cdots \exists \bar x_{i-1} \forall \bar x_{i} \, \psi$ (resp., 
$\forall \bar x_1 \exists \bar x_2 \cdots \forall \bar x_{i-1} \exists \bar x_{i} \, \psi$) if $i$ is even, and of the form
$\exists \bar x_1 \forall \bar x_2 \cdots \forall \bar x_{i-1} \exists \bar x_{i} \, \psi$ (resp., 
$\forall \bar x_1 \exists \bar x_2 \cdots \exists \bar x_{i-1} \forall \bar x_{i} \, \psi$) if $i$ is odd, where $\psi$ is a quantifier-free formula. Finally, we say that a fragment $\L_1$ is contained in a fragment $\L_2$, denoted by $\L_1 \subseteq \L_2$, if for every formula $\varphi$ in $\L_1$, there exists a formula $\psi$ in $\L_2$ such that $\varphi$ is logically equivalent to $\psi$.  Besides, we say that $\L_1$ is properly contained in $\L_2$, denoted by $\L_1 \subsetneq \L_2$, if $\L_1 \subseteq \L_2$ and $\L_2 \not\subseteq \L_1$.

\subsection{Counting complexity classes}

We consider several counting complexity classes in this paper, some of them are recalled here~(see \cite{F97,hemaspaandra2013complexity}).
$\fp$ is the class of functions $f \colon \Sigma^* \to \N$ computable in polynomial time, while $\fpspace$ is the class of functions $f \colon \Sigma^* \to \N$ computable in polynomial space. 
Given a nondeterministic Turing Machine (NTM) $M$, let $\tma_M(x)$ be the number of accepting runs of $M$ with input $x$. Then $\shp$ is the class of functions $f$ for which there exists a polynomial-time NTM $M$ such that $f(x) = \tma_M(x)$ for every input $x$, while $\shl$ is the class of functions $f$ for which there exists a logarithmic-space NTM $M$ such that $f(x) = \tma_M(x)$ for every input $x$.  Given an NTM $M$ with output tape, let $\tmo_M(x)$ be the number of distinct outputs of $M$ with input $x$ (notice that  $M$ produces an output if it halts in an accepting state). Then $\spp$ is the  class of functions $f$ for which there exists a polynomial-time NTM $M$ such that $f(x) = \tmo_M(x)$ for every input $x$. Notice that $\shp \subseteq \spp$, and this inclusion is believed to be strict. 





\section{A logic for quantitative functions} \label{sec:logic}

We introduce here the logical framework that we use for studying counting complexity classes. 
This framework is based on the framework of Weighted Logics (WL)~\cite{DrosteG07}  that has been used in the context of weighted automata for studying functions from words (or trees) to semirings. 
We propose here to use the framework of WL over any relational structure and to restrict the semiring to natural numbers. 
The extension to any relational structure will allow us to study general counting complexity classes and the restriction to the natural numbers will simplify the notation in this context (see Section~\ref{sec:previous} for a more detailed discussion).

Given a relational signature $\R$, the set of Quantitative Second-Order logic formulae (or just $\qso$-formulae) over $\R$ is given by the following grammar:
\begin{align}
\alpha \ &:= \ \varphi \ \mid \ s \ \mid \ (\alpha \add \alpha) \ \mid\ (\alpha \mult \alpha) \ \mid \sa{x} \alpha \ \mid \ \pa{x} \alpha \ \mid \ \sa{X} \alpha \ \mid \ \pa{X} \alpha \label{syntax} 
\end{align}
where $\varphi$ is an $\so$-formula over $\R$, $s \in \bbN$, $x \in \fv$ and $X \in \sv$. Moreover, if $\R$ is not mentioned, then $\qso$ refers to the set of $\qso$ formulae over all possible relational signatures.

The syntax of QSO formulae is divided in two levels. 
The first level is composed by $\so$-formulae over $\R$ (called Boolean formulae) and the second level is made by counting operators of addition and multiplication. 
For this reason, the quantifiers in $\so$ (e.g.\ $\exists x$ or $\exists X$) are called Boolean quantifiers and the quantifiers that make use of addition and multiplication (e.g.\ $\Sigma x$ or $\Pi X$) are called {\em quantitative quantifiers}.
Furthermore, $\Sigma x$ and $\Sigma X$ are called first- and second-order sum, whereas $\Pi x$ and $\Pi X$ are called first- and second-order product, respectively.
This separation between the Boolean and quantitative levels is essential for understanding the difference between the logic and the quantitative parts of the framework.
Furthermore, this will later allow us to parametrize both levels of the logic in order to capture different counting complexity classes.

\begin{table}
	\addtolength{\jot}{0.5em}
	\begin{align*}
	\sem{\varphi}(\A, v, V) & = 
	\begin{cases}
	1 & \mbox{if } (\A, v, V) \models \varphi \\
	0 & \mbox{otherwise}
	\end{cases}\\
	\sem{s}(\A, v, V) & = s \\
	\sem{\alpha_1 \add \alpha_2}(\A, v, V) & = \sem{\alpha_1}(\A, v, V) + \sem{\alpha_2}(\A, v, V)\\
	\sem{\alpha_1 \mult \alpha_2}(\A, v, V) & = \sem{\alpha_1}(\A, v, V) \cdot \sem{\alpha_2}(\A, v, V)\\ 
	\sem{\sa{x} \alpha}(\A, v, V) & = \displaystyle \sum_{a \in A} \sem{\alpha}(\A,v[a/x],V)\\
	\sem{\pa{x} \alpha}(\A, v, V) & = \displaystyle \prod_{a \in A} \sem{\alpha}(\A,v[a/x],V)\\
	\sem{\sa{X} \alpha}(\A, v, V) & = \displaystyle \sum_{B \subseteq A^{\arity(X)}} \sem{\alpha}(\A, v, V[B/X])\\
	\sem{\pa{X} \alpha}(\A, v, V) & = \displaystyle \prod_{B \subseteq A^{\arity(X)}} \sem{\alpha}(\A, v, V[B/X])
	\end{align*}
	\caption{The semantics of QSO formulae.}
	\label{tab-semantics}
\end{table}
Let $\R$ be a signature, $\A$ an $\R$-structure with domain $A$, $v$ a first-order assignment for $\A$ and $V$ a second-order assignment for $\A$. Then the \emph{evaluation} of a $\qso$-formula $\alpha$ over $(\A, v, V)$ is defined as a function $\sem{\alpha}$ that on input $(\A, v, V)$ returns a number in $\bbN$. Formally, the function $\sem{\alpha}$ is recursively defined in Table~\ref{tab-semantics}.
A $\qso$-formula $\alpha$ is said to be a \emph{sentence} if it does not have any free variable, that is, every variable in $\alpha$ is under the scope of a usual quantifier or a quantitative quantifier. It is important to notice that if $\alpha$ is a $\qso$-sentence over a signature $\R$, then for every $\R$-structure $\A$, first-order assignments $v_1$, $v_2$ for $\A$ and second-order assignments $V_1$, $V_2$ for $\A$, it holds that $\sem{\alpha}(\A, v_1, V_1) = \sem{\alpha}(\A, v_2, V_2)$.
Thus, in such a case we use the term $\sem{\alpha}(\A)$ to denote $\sem{\alpha}(\A, v, V)$, for some arbitrary first-order assignment $v$ for $\A$ and some arbitrary second-order assignment $V$ for $\A$. 
\begin{exa}\label{ex:cliques}
Let $\bG = \{E(\cdot,\cdot),<\}$ be the vocabulary for graphs and $\fG$ be an ordered $\bG$-structure encoding an undirected graph. 
Suppose that we want to count the number of triangles in $\fG$. Then this can be defined as follows:
\begin{align*}
\alpha_1 \ &:= \ \sa{x} \sa{y} \sa{z} ( E(x,y) \, \wedge \, E(y,z) \, \wedge \, E(z,x) \, \wedge x < y \, \wedge \, y < z )
\end{align*}
We encode a triangle in $\alpha_1$ as an increasing sequence of nodes $\{x, y, z\}$, in order to count each triangle once. Then the Boolean subformula  $E(x,y) \wedge E(y,z) \wedge E(z,x) \wedge
x < y \wedge y < z$ is checking the triangle property, by returning $1$ if $\{x, y, z\}$ forms a triangle in $\fG$ and $0$ otherwise.
Finally, the sum quantifiers in $\alpha_1$ aggregate all the values, counting the number of triangles in $\fG$.

Suppose now that we want to count the number of cliques in~$\fG$. We can define this function with the following formula:
\[
\alpha_2 \ := \ \sa{X} \clique(X),
\] 
where $\clique(X) := \fa{x} \fa{y} ((X(x) \wedge X(y) \wedge x \neq y)  \rightarrow E(x,y))$.
In the Boolean sub-formula of $\alpha_2$ we check whether $X$ is a clique, and with the sum quantifier we add one for each clique in $\fG$. 
But in contrast to $\alpha_1$, 
in $\alpha_2$ we need a second-order quantifier in the quantitative level.
This is according to the
complexity of evaluating each formula:
$\alpha_1$ defines an $\fp$-function while $\alpha_2$ defines a $\shp$-complete function. \qed
\end{exa}
\begin{exa}\label{exa-perm}
For an example that includes multiplication, let $\bM = \{M(\cdot,\cdot),<\}$ be a vocabulary for storing 0-1 matrices; in particular, a structure $\fM$ over $\bM$ encodes a 0-1 matrix $A$ as follows: if $A[i,j] = 1$, then $M(i,j)$ is true, otherwise $M(i,j)$ is false.
Suppose now that we want to compute the permanent of an $n$-by-$n$ 0-1 matrix $A$, that is:
\begin{align*}
\op{perm}(A) &= \sum_{\sigma \in S_n} \prod_{i=1}^n A[i, \sigma(i)],  
\end{align*}
where $S_n$ is the set of all permutations over $\{1, \ldots, n\}$.
The permanent is a fundamental function on matrices that has found many applications;
in fact, showing that this function is hard to compute was one of the main motivations behind the definition of the class $\shp$~\cite{Valiant79}.

To define the permanent of a 0-1 matrix in $\qso$, assume that for a binary relation symbol $S$, $\op{permut}(S)$ is an $\fo$-formula that is true if, and only if, $S$ is a permutation, namely, a total bijective function (the definition of $\op{permut}(S)$ is straightforward).
Then the following is a $\qso$-formula defining the permanent of a matrix:
\[
\alpha_3 := \sa{S} \op{permut}(S) \cdot \pa{x} (\ex{y} S(x,y) \wedge M(x,y)).
\]
Intuitively, the subformula $\beta(S) := \pa{x} (\ex{y} S(x,y) \wedge M(x,y))$ calculates the value \linebreak $\prod_{i=1}^n A[i, \sigma(i)]$ whenever $S$ encodes a permutation $\sigma$.
Moreover, the subformula $\op{permut}(S) \cdot \beta(S)$ returns $\beta(S)$ when $S$ is a permutation, and returns $0$ otherwise (i.e.\ $\op{permut}(S)$ behaves like a filter). 
Finally, the second order sum aggregates these values iterating over all binary relations and calculating the permanent of the matrix.
We would like to finish with this example by highlighting the similarity of $\alpha_3$ to the permanent formula. 
Indeed, an advantage of $\qso$-formulae is that the first- and second-order quantifiers in the quantitative level naturally reflect the operations used to define mathematical formulae. \qed
\end{exa}


\bigskip

\noindent
{\bf On restricting $\qso$.} 
We consider several fragments of $\qso$, which are obtained by restricting the use of the quantitative quantifiers or the syntax of the Boolean formulae, and we consider some extensions that are obtained by adding recursive operators to $\qso$.
\begin{itemize}
\item If we need to restrict the use of the quantitative quantifiers, then we replace $\qso$ by a term denoting the quantitative quantifiers that are allowed.
In particular, we denote by $\qfo$ the fragment of $\qso$ where only first-order sum and product are allowed. 
For instance, for the $\qso$-formulae defined in Example \ref{ex:cliques}, we have that $\alpha_1$ is in $\qfo$ and $\alpha_2$ is not.
Moreover, we denote by $\eqso$ the fragment of $\qso$ where only first- and second-order sums are allowed (that is, $\pa{x}$ and $\pa{X}$ are not allowed).
For example, $\alpha_1$ and $\alpha_2$ in Example \ref{ex:cliques} are formulae of $\eqso$, while $\alpha_3$ in Example \ref{exa-perm} is not. 

\item If we need to restrict the syntax of the Boolean formulae to a fragment $\LL$ of $\so$, then we use notation $\qso(\LL)$ to indicate that $\varphi$ in \eqref{syntax} has to be a formula in $\LL$. Moreover, every fragment of $\qso$ can be restricted by using the same idea. 
For example, $\qfo(\fo)$ is the fragment of $\qfo$ obtained by restricting $\varphi$ in \eqref{syntax} to be an $\fo$-formula, and likewise for $\eqso(\fo)$.
\end{itemize}

\bigskip

\noindent
In the following section, we use different fragments or extensions of $\qso$ to capture counting complexity classes. But before doing this, we show the connection of $\qso$ to previous frameworks for defining functions over relational structures.

\subsection{Previous frameworks for quantitative functions} \label{sec:previous}

In this section, we discuss some previous frameworks proposed in the literature and how they differ from our approach.
We start by discussing the connection between $\qso$ and weighted logics (WL)~\cite{DrosteG07}. 
As it was previously discussed, $\qso$ is a fragment of WL.
The main difference is that we restrict the semiring used in WL to natural numbers in order to study counting complexity classes.
Another difference between WL and our approach is that, to the best of our knowledge, this is the first paper to study weighted logics over general relational signatures, in order  to do descriptive complexity for counting complexity classes. 
Previous works on WL usually restrict the signature of the logic to strings, trees, and other specific structures (see \cite{droste2009handbook} for more examples), and they did not study the logic over general structures. 
Furthermore, in this paper we propose further extensions for $\qso$ (see Section~\ref{sec:beyond}) which differ from previous approaches in WL.

Another approach that resembles $\qso$ are logics with counting~\cite{Immerman82,IL90,E97,GG98,L04}, which include operators that extend $\fo$ with quantifiers that allow to count in how many ways a formula  is satisfied (the result of this counting is a value of a second sort, in this case the  natural numbers). 
A particularly influential logic in the area was proposed in \cite{Immerman82,IL90}, which is usually referred as FPC, and it is defined as an extension of first-order logic with a least fixed-point operator and counting quantifiers.
However, and in contrast with our approach, counting operators are usually used in these logics for checking Boolean properties over structures and not for producing values (i.e.\ they do not define a function).
In particular, we are not aware of any paper that uses this approach for capturing counting complexity classes.

Finally, the work in~\cite{SalujaST95,ComptonG96,0001HKV16} is of particular interest for our research. 
In~\cite{SalujaST95}, it was proposed to define a function over a structure by using free variables in an SO-formula; in particular, the function is defined by the number of instantiations of the free variables that are satisfied by the structure.
Formally, Saluja et. al \cite{SalujaST95} define a family of counting classes $\#\LL$ for a fragment $\LL$ of $\fo$. For a formula $\varphi(\bar{x},\bar{X})$ over $\R$, the function $f_{\varphi(\bar x, \bar X)}$ is defined as
$
f_{\varphi(\bar x, \bar X)}(\A) = \vert \{(\bar{a},\bar{A}) \mid \A\models\varphi(\bar{a},\bar{A})\}\vert
$
for every $\A\in\ostr[\R]$. Then a function $g\colon \ostr[\R]\to\nat$ is in $\#\LL$ if there exists a formula $\varphi(\bar{x},\bar{X})$ in $\LL$ such that $g = f_{\varphi(\bar x, \bar X)}$.
In~\cite{SalujaST95}, they proved several results about capturing counting complexity classes which are relevant for our work. We discuss and use these results in Sections~\ref{sec:complexity} and~\ref{sec:syntactic}.
Notice that for every formula $\varphi(\bar{x},\bar{X})$, it holds that $f_{\varphi(\bar{x},\bar{X})}$ is the same function as $\sem{\sa{\bar{X}} \sa{\bar{x}} \varphi(\bar{x},\bar{X})}$, that is, the approach in \cite{SalujaST95} can be seen as a syntactical restriction of our approach based on $\qso$. 
Thus, the advantage of our approach relies on the flexibility to define functions by alternating sum with product operators and, moreover, by introducing new quantitative operators (see Section~\ref{sec:beyond}).
Furthermore, we show in the next section how to capture some classes that cannot be captured by following the approach in~\cite{SalujaST95}.


\section{Counting under $\qso$} \label{sec:complexity}

In this section, we show that by syntactically restricting $\qso$ one can capture different counting complexity classes. 
In other words, by using $\qso$ we can extend the theory of descriptive complexity~\cite{immerman1999descriptive} from decision problems to counting problems. 
For this, we first formalize the notion of \emph{capturing} a complexity class of functions.

Fix a signature $\R = \{R_1, \ldots, R_k\}$ and assume that $\A$ is an ordered $\R$-structure with a domain $A = \{a_1, \ldots, a_n\}$, $R_k =\, <$, and $a_1 <^{\A} a_2 <^{\A} \ldots <^{\A} a_n$. For every $i \in \{1, \ldots, k-1\}$, define the encoding of $R_i^\A$, denoted by $\enc(R_i^\A)$, as the following binary string. Assume that $\ell = \arity(R_i)$ and consider an enumeration of the $\ell$-tuples over $A$ in the lexicographic order induced by $<$. 
Then let $\enc(R_i^\A)$ be a binary string of length $n^\ell$ such that the $i$-th bit of $\enc(R_i^\A)$ is 1 if the $i$-th tuple in the previous enumeration belongs to $R_i^\A$, and 0 otherwise. Moreover, define the encoding of $\A$, denoted by $\enc(\A)$, as the string~\cite{L04}:
$$
0^n \, 1 \, \enc(R_1^\A) \, \cdots \, \enc(R_{k-1}^\A).
$$

We can now formalize the notion of capturing a counting complexity class.
\begin{defi} \label{def:cap}
	Let $\FF$ be a fragment of $\qso$ and $\CC$ a counting complexity class. Then {\em  $\FF$ captures $\CC$ over ordered $\R$-structures} if the  following conditions hold:
	\begin{enumerate}
		\item for every $\alpha \in \FF$, there exists $f \in \CC$ such that $\sem{\alpha}(\A) = f(\enc(\A))$ for every $\A \in \ostr[\R]$. 
		
		\item for every $f \in \CC$, there exists $\alpha \in \FF$ such that   $f(\enc(\A)) = \sem{\alpha}(\A)$ for every $\A \in \ostr[\R]$.
	\end{enumerate} 
	Moreover, {\em $\FF$ captures $\CC$ over ordered structures} if $\FF$ captures~$\CC$ over ordered $\R$-structures for every signature~$\R$. \qed
\end{defi}
In Definition~\ref{def:cap}, function $f \in \CC$ and formula $\alpha \in \FF$ must coincide in all the strings that encode ordered $\R$-structures. Notice that this restriction is natural as we want to capture 
$\CC$ over a fixed set of structures (e.g.\ graphs, matrices).
Moreover, this restriction is fairly standard in descriptive complexity \cite{immerman1999descriptive,L04}, and it has also been used in the previous work on capturing complexity classes of functions \cite{SalujaST95,ComptonG96}.

What counting complexity classes can be captured with fragments of $\qso$?
For answering this question, it is reasonable to start with $\shp$, a well-known and widely-studied counting complexity class~\cite{arora2009computational}. 
Since $\shp$ has a strong similarity with $\np$, one could expect a ``Fagin-like'' Theorem~\cite{F75} for this class. 
Actually, in~\cite{SalujaST95} it was shown that the class $\sfo$ captures $\shp$.
In our setting, the class $\sfo$ is contained in $\eqso(\fo)$, which also captures $\shp$ as expected.
 
\begin{prop} \label{prop:capture-shP}
	$\eqso(\fo)$ captures $\shp$ over ordered structures.
\end{prop}
\begin{proof}
We briefly explain how the two conditions of Definition~\ref{def:cap} are satisfied. First, for condition (2) Saluja et al. proved that $\shp = \sfo$\cite{SalujaST95}. Hence, given that every function in $\sfo$ can be trivially defined as a formula in $\eqso(\fo)$ (see Section~\ref{sec:previous}), condition~(2) holds.
For condition (1), let $\alpha\in\eqso(\fo)$ over some signature $\R$. Given an $\fo$ formula $\varphi$, checking whether $\A\models\varphi$ can be done in deterministic polynomial time on the size of $\A$ and any constant function $s$ can be trivially simulated in $\shp$. These facts, together with the closures under exponential sum and polynomial product of $\shp$\cite{F97}, suffice to show that the function represented by $\alpha$ is in $\shp$.
%
 
\end{proof}

By following the same approach as~\cite{SalujaST95}, Compton and Gr\"adel~\cite{ComptonG96} show that $\seso$ captures $\spp$, where $\eso$ is the existential fragment of $\so$. As one could expect, if we parametrize $\eqso$ with $\eso$, we can also capture~$\spp$.
\begin{prop} \label{prop:capture-spanP}
	$\eqso(\eso)$ captures $\spp$ over ordered structures.
\end{prop}
\begin{proof}
To prove condition (2), we use the fact that $\spp = \#(\eso)$. The condition holds using the same argument as in Proposition~\ref{prop:capture-shP}.
For condition (1), notice that given an $\eso$ formula $\varphi$, checking whether $\A\models\varphi$ can be done in non-deterministic polynomial time on the size of $\A$\cite{fagin1974generalized}. 
Therefore, a $\spp$ machine for $\varphi$ will simulate the non-deterministic polynomial time machine and produce the same string as output in each accepting non-deterministic run. Furthermore, any constant function $s$ can be trivially simulated in $\spp$ and, thus, condition (1) holds analogously to Proposition \ref{prop:capture-shP} since $\spp$ is also closed under exponential sum and polynomial product~\cite{OH93}.

\end{proof}

Can we capture $\fp$ by using $\# \LL$ for some fragment $\LL$ of $\so$? A first attempt could be based on the use of a fragment $\LL$ of $\so$ that captures either $\ptime$ or $\nlog$~\cite{G92}. Such an approach fails as $\# \LL$ can encode $\shp$-complete problems in both cases; in the first case, one can encode the problem of counting the number of satisfying assignments of a Horn  propositional formula, while in the second case one can encode the problem of counting the number of satisfying assignments of a 2-CNF propositional formula. A second attempt could then be based on considering a fragment $\LL$ of $\fo$. 
But even if we consider the existential fragment $\Sigma_1$ of $\fo$ the approach fails, as $\# \Sigma_1$ can encode $\shp$-complete problems like counting the number of satisfying assignments of a 3-DNF propositional formula\cite{SalujaST95}. One last attempt could be based on disallowing the use of second-order free variables in $\sfo$. But in this case one 
cannot capture exponential functions definable in $\fp$ such as~$2^n$.
Thus, it is not clear how to capture $\fp$ 
by following the approach proposed in~\cite{SalujaST95}. 
On the other hand, if we consider our framework and move out from $\eqso$, we have other alternatives for counting like first- and second-order products. In fact, the combination of $\qfo$ with $\lfp$ is exactly what we need to capture $\fp$.
\begin{thm} \label{theo:capture-fp}
	$\qfo(\lfp)$ captures $\fp$ over ordered structures.
\end{thm}
\begin{proof}
In this and the following proofs, we will reuse the symbol $<$ to denote the lexicographic order over same-sized tuples. 
Formally, for $\bar{x} = (x_1,\ldots,x_m)$ and $\bar{y} = (y_1,\ldots,y_m)$ we denote by $\bar{x} < \bar{y}$ the formula:
$$
\bigvee_{i = 1}^m \, \bigwedge_{j = 1}^{i-1} \, (x_j = y_j \wedge x_i < y_i).
$$
Similarly, we use $\bar{x} = \bar{y}$ to denote equality between tuples and $\bar{x} \leq \bar{y}$  to denote $\bar{x} < \bar{y} \vee \bar{x} = \bar{y}$.
We will also use some syntactic sugar in $\qso$ to simplify formulae. 
Specifically, we will use the {\em conditional count} symbol $(\varphi \mapsto \alpha)$ defined as $ (\varphi\cdot\alpha) + \neg\varphi$ for any Boolean formula $\varphi$ and any quantitative formula $\alpha$. 
Note that for each $\A \in \ostr[\R]$, and each first-order (second-order) assignment $v$ ($V$) over $\A$:
\[
\sem{(\varphi \mapsto \alpha)}(\A,v,V) = 
\begin{cases}
\sem{\alpha}(\A,v,V) &\text{if } (\A,v,V)\models\varphi,\\
1 &\text{otherwise}.
\end{cases}
\]
Furthermore, we use $\size{\A}$ to denote the size of an $\R$-structure $\A$.
Now we prove Theorem~\ref{theo:capture-fp}. 
For condition (1), recall that checking whether $\A\models\varphi$ for any $\lfp$ formula $\varphi$ can be done in deterministic polynomial time on the size of $\A$\cite{I83}. 
Furthermore, it is easy to check that $\fp$ is closed under polynomial sum and multiplication. 
We conclude then that any formula in $\qfo(\lfp)$ can be computed in $\fp$.	
For condition~(2), let $\R$ be a signature, $f\in \fp$
and $\ell\in\nat$ such that $\log_2\left(f(\enc(\A)) \right) \leq \size{\A}^\ell$ for every $\A\in\ostr[\R]$ (i.e.\ $\size{\A}^\ell$ is an upper bound for the output size of $f$ over $\A$).
Consider the language:
\[
L = \{(\A,\bar{a})\mid \bar{a} \in A^l \text{ and the } \bar{a}\text{-th bit of }f(\enc(\A))\text{ is 1}\}.
\]
where $\bar{a}$ encodes a number by following the lexicographic order over $A^l$.
Clearly, the language $L$ is in $\ptime$ and by \cite{I83} there exists a formula $\Phi(\bar{x})$ in $\lfp$ such that $\A\models\Phi(\bar{a})$ if, and only if, $(\A,\bar{a})\in L$. 
We use then the following formula to encode $f$:
$$
\alpha = \sa{\bar{x}} \Phi(\bar{x})\cdot \pa{\bar{y}}(\bar{y} < \bar{x}) \mapsto 2)$$
Note that the subformula $\pa{\bar{y}}(\bar{y} < \bar{x}) \mapsto 2$ takes the value $2^m$ if there exist $m$ tuples in $A^{\ell}$ that are smaller than $\bar{x}$. Adding these values for each $\bar{a}\in A^{\ell}$ gives exactly $f(\enc(\A))$. 
In other words, $\Phi(\bar{x})$ simulates the behavior of the $\fp$-machine and the formula $\alpha$ reconstructs the binary output bit by bit.
Then $\alpha$ is in $\qfo(\lfp)$ and $\sem{\alpha}(\A) = f(\enc(\A))$.

\end{proof}

At this point it is natural to ask whether one can extend the previous idea to capture $\fpspace$~\cite{Ladner89}, the class of functions computable in polynomial space. 
Of course, for capturing this class one needs a logical core powerful enough, like $\pfp$, for simulating the run of a polynomial-space TM.
Moreover, 
one also needs more powerful quantitative quantifiers as functions like $2^{2^n}$ can be computed in polynomial space,
so $\eqso$ is not enough for the quantitative layer of a logic for $\fpspace$.
In fact, by considering second-order product we obtain the fragment $\qso(\pfp)$ that captures $\fpspace$. 
\begin{thm} \label{theo:capture-fpspace}
	$\qso(\pfp)$ captures $\fpspace$ over ordered structures.
\end{thm}
\begin{proof}

For the first condition of Definition~\ref{def:cap}, notice that each $\pfp$ formula can be evaluated in deterministic polynomial space, the constant function $s$ can be trivially simulated in $\fpspace$, and $\fpspace$ is closed under exponential sum and multiplication. This suffices to show that the condition holds.
For the second condition, the proof is similar to the proof of Theorem~\ref{theo:capture-fp}. Let $f\in \fpspace$ defined over some $\R$ and $\ell\in\nat$ such that $\log_2\left( f(\enc(\A)) \right) \leq 2^{{|\A|}^\ell}$ for every $\A\in\ostr[\R]$  (i.e.\ $2^{{|\A|}^\ell}$ is an upper bound for the size of the output). Let $X$ be a second-order variable of arity $\ell$. Consider the linear order induced by $<$ over predicates of arity $\ell$ which can be defined by the following formula:
$$
\varphi_{<}(X,Y) = \ex{\bar{u}}\big[\neg X(\bar{u})\wedge Y(\bar{u})\wedge \fa{\bar{v}}\big(
\bar{u}<\bar{v}\to(X(\bar{u})\iff Y(\bar{v}))\big)\big].
$$
Namely, we use predicates to encode a number that will have most $2^{{|\A|}^\ell}$ bits. We define this encoding through the function $\tau\colon 2^{A^\ell}\to\nat$, such that $\tau(B)$ is equal to the number of predicates in $2^{A^\ell}$ that are smaller than $B$ with respect to the induced order. For example, we have that $\tau(\emptyset) = 0$ and $\tau(A^{\ell}) = 2^{{|\A|}^\ell}-1$. Furthermore, we can use a relation~$X$ to index a position in the binary output of $f(\enc(\A))$ as follows.
Define the language:
\[
L = \{(\A,B)\mid B \subseteq A^{\ell}\text{ and the $\tau(B)$-th bit of $f(\enc(\A))$ is 1}\}.
\]
Since $L$ is in $\pspace$, it can be specified in $\pfp$ \cite{AbiteboulV89} by a formula $\Phi(X)$ such that $\A\models\Phi(B)$ if and only if $(\A,B)\in L$. Then, similarly as for the previous proof we define:
$$
\alpha := \sa{X} \Phi(X)\mult  \pa{Y}(\varphi_{<}(Y,X)\mapsto 2).
$$ 
where $\pa{Y}(\varphi_{<}(Y,X)\mapsto 2)$ takes the value $2^{\tau(X)}$ and $\alpha$ reconstructs the output of $f(\enc(\A))$. Using an argument analogous to the previous proof, we conclude that $\alpha\in\qso(\pfp)$ and $\sem{\alpha}(\A) = f(\enc(\A))$.
%

\end{proof}

From the proof of the previous theorem a natural question follows: what happens if we use first-order quantitative quantifiers and $\pfp$?
In~\cite{Ladner89}, Ladner also introduced the class $\nfpspace$ of all functions computed by polynomial-space TMs 
with output length bounded by a polynomial.
Interestingly, if we restrict to FO-quantitative quantifiers we can also capture this class.
\begin{cor} \label{cor:capture-fpspace-poly}
	$\qfo(\pfp)$ captures $\nfpspace$ over ordered structures.
\end{cor}
\begin{proof}
In this proof, both conditions are analogous to Theorem~\ref{theo:capture-fp} and~\ref{theo:capture-fpspace}. For the first condition, each $\pfp$ formula $\varphi$ can be evaluated in $\pspace$ and the class is closed under first sum and product. For the second condition, we use the same language $L$ defined in the proof of Theorem~\ref{theo:capture-fp}, which in this case is in $\pspace$. The same construction of $\alpha$, which in this case is in $\qfo(\pfp)$, is used to show that the condition holds.

\end{proof}

The results of this section validate $\qso$ as an appropriate logical framework for extending the theory of descriptive complexity to counting complexity classes. In the following sections, we provide more arguments for this claim, by considering some fragments of $\eqso$ and, moreover, by showing how to go beyond $\eqso$ to capture other classes.


\subsection{Extending $\qso$ to capture classes beyond counting} \label{sec:extentions}

There exist complexity classes that do not fit in our framework because either the output of a function is not a natural number (e.g.\ a negative number) or the class is not defined purely in terms of arithmetical operations (e.g.\ min and max).
To remedy this problem, we show here how $\qso$ can be easily extended  to capture such classes that go beyond sum and product over natural numbers. 

It is well-known that, under some reasonable complexity-theoretical assumptions, $\shp$ is not closed under subtraction, not even under subtraction by one~\cite{OH93}.
To overcome this limitation, $\gp$ was introduced in~\cite{FFK94} as the class of functions $f$ for which there exists a polynomial-time NTM $M$ such that $f(x) = \tma_M(x) - \tmr_M(x)$, where  $\tmr_M(x)$ is the number of rejecting runs of $M$ with input $x$.
That is, $\gp$ is the closure of $\shp$ functions under subtraction, and its functions can obviously take negative values.
Given that our logical framework was built on top of the natural numbers, we need to extend $\qso$ in order to capture $\gp$. 
The most elegant way to do this is by allowing constants coming from $\bbZ$ instead of just $\bbN$. 
Formally, we define the logic $\qsoz$ whose syntax is the same as in \eqref{syntax} and whose semantics is the same as in Table~\ref{tab-semantics} except that the atomic formula $s$ (i.e.\ a constant) comes from $\bbZ$.  
Similarly as for $\qso$, we define the fragment $\eqsoz$ as the extension of $\eqso$ with constants in $\bbZ$.
\begin{exa}
	Recall the setting of Example~\ref{ex:cliques} and suppose now that we want to compute the number of cliques in a graph that are not triangles. This can be easily done in $\qsoz$ with the formula:
	$
	\alpha_5 :=	\alpha_2 + (-1) \cdot \alpha_1. 
	$ \qed
\end{exa}
Adding negative constants is a mild extension to allow subtraction in the logic. 
It follows from our characterization of $\shp$ that this is exactly what we need to capture  $\gp$.
\begin{cor} \label{prop:capture-gapp}
	$\eqsoz(\fo)$ captures $\gp$ over ordered structures.
\end{cor}
This is an interesting result that shows how robust and versatile $\qso$ is for capturing different counting complexity classes whose functions are not restricted to $\bbN$.

A different class of functions comes from considering the optimization version of a decision problem. For example, one can define MAX-SAT as the problem of determining the maximum number of clauses, of a given CNF propositional formula that can be made true by an assignment. Here, MAX-SAT is defined in terms of a maximization problem which in its essence differs from the functions in $\shp$. 
To formalize this class of optimization problems, Krentel defined $\optp$~\cite{krentel1988complexity} as the class of functions computable by taking the maximum or minimum of the output values over all runs of a polynomial-time NTM machine with output tape (i.e.\ each run produces a binary string which is interpreted as a number). 
For instance, MAX-SAT is in $\optp$ as many other optimization versions of $\np$-problems.
Given that in~\cite{krentel1988complexity} Krentel did not make the distinction between $\max$ and $\min$, Vollmer and Wagner~\cite{vollmer1995complexity} defined the classes $\maxp$ and $\minp$ as the max and min version of the problems in $\optp$ (i.e.\ $\optp = \maxp \cup \minp$).

In order to capture classes of optimization functions, we extend $\qso$ with $\max$ and $\min$ quantifiers as follows (called $\optqso$). 
Given a signature $\R$, the set of $\optqso$-formulae over $\R$ is given by extending the syntax in (\ref{syntax}) with the following operators:
\begin{align*}
\max\{\alpha,\alpha\} \ \mid\ \min\{\alpha,\alpha\} \ \mid \maxa{x} \alpha \ \mid \ \mina{x} \alpha \ \mid \ \maxa{X} \alpha \ \mid \ \mina{X} \alpha 
\end{align*}
where $x \in \fv$ and $X \in \sv$. The semantics of the $\qso$-operators in $\optqso$ are defined as usual. Furthermore, the semantics of the max and min quantifiers are defined as the natural extension of the sum quantifiers in $\qso$ (see Table~\ref{tab-semantics}) by maximizing or minimizing, respectively, instead of computing a sum or a product. 
\begin{exa}\label{ex:optqso}
	Recall again the setting of Example~\ref{ex:cliques} and suppose now that we want to compute the size of the largest clique in a graph. This can be done in $\optqso$ as follows:
	\[
\alpha_6 := \maxa{X} \left( \, \clique(X) \cdot \sa{z} X(z)  \, \right)
	\]
	Notice that formula $\sa{z} X(z)$ is used to compute the number of nodes in a set $X$.  \qed
\end{exa}
Similarly as for $\maxp$ and $\minp$, we have to distinguish between the $\max$ and $\min$ fragments of $\optqso$. For this, we define the fragment $\maxqso$ of all $\optqso$ formulae constructed from $\qfo$ operators and $\max$-formulae $\max\{\alpha,\alpha\}$, $\maxa{x} \alpha$ and  $\maxa{X} \alpha$.
The class $\minqso$ is defined analogously replacing $\max$ by $\min$. Notice that in $\maxqso$ and $\minqso$, second-order sum and product are not allowed. For instance, formula $\alpha_6$ in Example \ref{ex:optqso} is in $\maxqso$.
As one could expect, $\maxqso$ and $\minqso$ are the needed logics to capture $\maxp$ and $\minp$.
\begin{thm} \label{theo:capture-optp}
	$\maxqso(\fo)$ and $\minqso(\fo)$ capture $\maxp$ and $\minp$, respectively, over ordered structures.
\end{thm}
\begin{proof}
It is straightforward to prove that $\maxp$ can compute any $\fo$-formula, is closed under first-order sum and product, and second-order maximization. 
Therefore, condition (1) in Definition~\ref{def:cap} follows similarly as in the previous characterizations. Furthermore, one can easily see that the same holds for $\minqso(\fo)$.
The proof for the other direction is similar to the one described in \cite{kolaitis1994logical} extended with the ideas of Theorem~\ref{theo:capture-fp}. Let $f\in \maxp$ be a function defined over some signature $\R$ and
$\ell\in\nat$ such that $\lceil\log_2 f(\enc(\A)) \rceil \leq |\A|^\ell$ for each $\A\in\ostr[\R]$.
For $U \subseteq A^{\ell}$, we can interpret the encoding of $U$ ($\enc(U)$) as the binary encoding of a number with $|\A|^\ell$-bits. We denote this value by $\val(\enc(U))$.
Then, given $\A\in\ostr[\R]$ and $U \subseteq A^{\ell}$, consider the problem of checking whether $f(\enc(\A)) \geq \val(\enc(U))$. 
Clearly, this is an $\np$-problem and, by Fagin's theorem, there exists a formula of the form $\ex{\bar{X}} \Phi(\bar{X}, Y)$ with $\Phi(\bar{X}, Y)$ in $\fo$ and $\arity(Y) = \ell$ such that $f(\enc(\A)) \geq \val(\enc(U))$ if, and only if, $(\A,v,V) \models \ex{\bar{X}} \Phi(\bar{X}, Y)$ with $V(Y) = U$. 
Then we can describe $f$ by the following $\maxqso$ formula:
$$
\alpha = \maxa{\bar{X}} \maxa{Y} \ \Phi(\bar{X}, Y) \cdot \big( \sa{\bar{x}} Y(\bar{x}) \cdot \pa{\bar{y}}(\bar{x} < \bar{y} \mapsto 2) \big).
$$
Note that, in contrast to previous proofs, we use $\bar{x} < \bar{y}$ instead of $\bar{y} < \bar{x}$ because the most significant bit in $\enc(U)$ correspond to the smallest tuple in $U$.  
It is easy to check that $\Phi(\bar{X}, Y)$ simulates the NP-machine and, if $\Phi(\bar{X}, Y)$ holds, the formula to the right  reconstructs the binary output from the relation in $Y$.
Then, $\alpha$ is in $\maxqso(\fo)$ over $\R$ and $\sem{\alpha}(\A) = f(\enc(\A))$. 

For the case of $\minqso(\fo)$ and a function $f \in \minp$, one has to follow the same approach but consider the $\np$-problem of checking whether $f(\enc(\A)) \leq \val(\enc(U))$. Then, the formula for describing $f$ is the following:
$$
\alpha = \mina{\bar{X}} \mina{Y} \ \sa{\bar{x}} \big( (\Phi(\bar{X}, Y) \rightarrow Y(\bar{x})) \cdot \pa{\bar{y}}(\bar{x} < \bar{y} \mapsto 2)  \big).
$$
In this case, if the formula $\Phi(\bar{X}, Y)$ is false, then the output produced by the subformula inside the $\min$-quantifiers will be the biggest possible value (i.e.\ $2^{{|\A|}^\ell}$).
On the other hand, if $\Phi(\bar{X}, Y)$ holds, the subformula will produce $\val(\enc(U))$. 
In a similar way as in $\max$, we conclude that $\alpha$ is in $\minqso(\fo)$ and $\sem{\alpha}(\A) = f(\enc(\A))$.
%

\end{proof}

It is important to mention that a similar result, following the framework of~\cite{SalujaST95}, was proved in~\cite{kolaitis1994logical} for the class $\maxpb$ (resp., $\minpb$) of problems in $\maxp$ (resp., $\minp$) whose output value is polynomially bounded.
Interestingly, Theorem \ref{theo:capture-optp} is stronger since our logic has the freedom to use sum and product quantifiers, instead of using a max-and-count problem over Boolean formulae. 
Finally, it is easy to prove that our framework can also capture $\maxpb$ and $\minpb$ by disallowing the product $\Pi x$ in $\maxqso(\fo)$ and $\minqso(\fo)$, respectively.
%

\section{Exploring the structure of $\shp$ through $\qso$} \label{sec:syntactic}

The class $\shp$ was introduced in \cite{Valiant79} to prove that computing the permanent of a matrix, as defined in Example \ref{exa-perm}, is a $\shp$-complete problem.
As a consequence of this result many counting problems have been proved to be $\shp$-complete~\cite{V79b,arora2009computational}.
Among them, problems having easy decision counterparts play a fundamental role, as a counting problem with a hard decision version is expected to be hard. Formally, the decision problem associated to a function $f\colon\Sigma^* \to \mathbb{N}$ is defined as $L_f = \{ x \in \Sigma^* \mid f(x) > 0 \}$, and $f$ is said to have an easy decision version if $L_f \in \ptime$. 
Many prominent examples satisfy this property, like computing the number of: perfect matchings of a bipartite graph ($\cpm$)~\cite{Valiant79}, satisfying assignments of a DNF ($\cdnf$)~\cite{DHK05,KL83} or Horn ($\chsat$)~\cite{V79b} propositional formula, among others.

Counting problems with easy decision versions play a fundamental role in the search for efficient approximation algorithms for functions in $\shp$. 
A fully-polynomial randomized approximation scheme (FPRAS) for a function $f\colon \Sigma^* \to \bbN$ is a randomized algorithm $\mathcal{A} \colon \Sigma^* \times (0,1) \to \bbN$ such that: (1) for every string $x \in \Sigma^*$ and real value $\varepsilon \in (0,1)$, the probability that $|f(x) - \mathcal{A}(x,\varepsilon)| \leq \varepsilon \cdot f(x)$ is at least $\frac{3}{4}$, and (2) the running time of $\mathcal{A}$ is polynomial in the size of $x$ and $1/\varepsilon$ \cite{KL83}. 
Notably, there exist $\shp$-complete functions that can be efficiently approximated as they admit FPRAS; for instance, there exist FPRAS for $\cdnf$~\cite{KL83} and $\cpm$~\cite{JSV04}. 
A key observation here is that if a function $f$ admits an FPRAS, then $L_f$ is in the randomized complexity class $\bpp$~\cite{G77}.
Hence, under the widely believed assumption that  $\np \not\subseteq \bpp$, we cannot hope for an FPRAS for a function in $\shp$ whose decision counterpart is $\np$-complete, and we have to concentrate on the class of counting problems with easy decision versions. That is, with decision versions in P.

The importance of the class of counting problems with easy decision counterparts has motivated the search for robust classes of functions in $\shp$ with this property \cite{PagourtzisZ06}. But the key question here is what should be considered a {\em robust} class. 
A first desirable condition is related to its closure properties, which is a common theme when studying function complexity classes \cite{OH93,FH08}. Analogously to the cases of $\ptime$ and $\np$, which are closed under intersection and union, we expect our class to be closed under multiplication and sum. For a more elaborated closure property, assume that $\textit{sat\_one}$ is a function that returns one plus the number of satisfying assignments of a propositional formula. Clearly $\textit{sat\_one}$ is a $\shp$-complete function whose decision counterpart $L_{\textit{sat\_one}}$ is trivial. But should $\textit{sat\_one}$ be part of a robust class of counting functions with easy decision versions? The key insight here is that if a function in $\shp$ has an easy decision counterpart $L$, then as $L \in \np$ we expect  to have a polynomial-time algorithm that verifies whether $x \in L$ by constructing witnesses for~$x$. 
Moreover, if such an algorithm for constructing witnesses exists, then we also expect to be able to manipulate such witnesses and in some cases to remove them. In other words, we expect a robust class $\CC$ of counting functions with easy decision versions to be closed under subtraction by one, that is, if $g \in \CC$, then the function $g \dotdiv 1$ should also be in $\CC$, where $(g \dotdiv 1)(x)$ is defined as $g(x) - 1$ if $g(x) \geq 1$, and as~$0$ otherwise. Notice that, unless $\ptime = \np$, no such class can contain the function $\textit{sat\_one}$ because $\textit{sat\_one} \dotdiv 1$ counts the number of satisfying assignments of a propositional formula.

A second desirable condition of robustness is the existence of natural complete problems~\cite{P94}. Special attention has to be paid here to the notion of reduction used for completeness. Notice that under the notion of Cook reduction, originally used in \cite{Valiant79}, the problems $\cdnf$ and $\csat$ are $\shp$-complete. However, $\cdnf$ has an easy decision counterpart and admits an FPRAS, while $\csat$ does not satisfy these conditions unless $\ptime = \np$. Hence a more strict notion of reduction has to be considered; in particular, the notion of parsimonious reduction (to be defined later) satisfies that if a function $f$ is parsimoniously reducible to a function $g$, then $L_g \in \ptime$ implies that $L_f \in \ptime$ and the existence of an FPRAS for $g$ implies the existence of a FPRAS for $f$. 

In this section, we use the framework developed in this paper to address the problem of defining a robust class of functions with easy decision versions. More specifically, we use the framework to introduce in Section \ref{sec-hier-shp} a syntactic hierarchy of counting complexity classes contained in $\shp$. Then this hierarchy is used in Section \ref{sec-clo} to define a class of functions with easy decision versions and good closure properties, and in Section \ref{sec-horn} to define a class of functions with easy decision versions and 
natural complete problems.

\subsection{The $\eqso(\fo)$-hierarchy inside $\shp$}
\label{sec-hier-shp}

Inspired by the connection between $\shp$ and $\sfo$, a hierarchy of subclasses of $\sfo$ was introduced in~\cite{SalujaST95} 
by restricting the alternation of quantifiers in Boolean formulae.
Specifically, the \emph{$\sfo$-hierarchy} consists of the 
the classes $\E{i}$ and $\U{i}$ for every $i \geq 0$, where $\E{i}$ (resp., $\U{i}$) is defined as $\sfo$ but restricting the formulae used to be in $\loge{i}$ (resp., $\logu{i}$).
By definition, we have that $\U{0} = \E{0}$. Moreover, it is shown in~\cite{SalujaST95} that:
\[
\E{0} \; \subsetneq \; \E{1} \; \subsetneq \; \U{1} \; \subsetneq \; \E{2} \; \subsetneq \; \U{2} \; = \; \sfo 
\]
In light of the framework introduced in this paper, natural extensions of these classes are obtained by considering 
$\eqso(\loge{i})$ and $\eqso(\logu{i})$ for every $i \geq 0$, which form the \emph{$\eqso(\fo)$-hierarchy}.
Clearly, we have that $\E{i} \subseteq \QE{i}$ and $\U{i} \subseteq \QU{i}$. Indeed, each formula $\varphi(\bar{X}, \bar{x})$ in $\E{i}$ is equivalent to the formula $\sa{\bar X} \sa{\bar x} \varphi(\bar{X}, \bar{x})$ in $\QE{i}$, and likewise for $\U{i}$ and $\QU{i}$.
But what is the exact relationship between these two hierarchies?
To answer this question, we first introduce two normal forms for $\eqso(\LL)$ that helps us to characterize the expressive power of this quantitative logic.
A formula $\alpha$ in $\eqso(\LL)$ is in \emph{$\LL$-prenex normal form} ($\LL$-PNF) 
if $\alpha$ is of the form
$\sa{\bar{X}} \sa{\bar{x}} \varphi(\bar{X}, \bar{x})$,
where $\bar{X}$ and $\bar{x}$ are sequences of zero or more second-order and first-order variables, respectively, (as expected, $\sa{\bar{X}}\!$ and $\sa{\bar{x}}\!$ are the respective nestings of $\sa{X}$'s and $\sa{x}$'s) and $\varphi(\bar{X}, \bar{x})$ is a formula in $\LL$. Notice that 
a formula $\varphi(\bar{X}, \bar{x})$ in $\sh{\LL}$ is equivalent to the formula $\sa{\bar X} \sa{\bar x} \varphi(\bar{X}, \bar{x})$ in $\LL$-PNF. 
Moreover, a formula $\alpha$ in $\eqso(\LL)$ is in \emph{$\LL$-sum normal form} ($\LL$-SNF) if $\alpha$ is of the form $\sum_{i=1}^n \alpha_i$, where this is a shorthand notation for $\alpha_1+\cdots+\alpha_n$, and each $\alpha_i$ is in $\LL$-PNF. 
\begin{prop}\label{theo-pnf-snf}
Every formula in $\eqso(\LL)$ can be rewritten in $\LL$-SNF.
\end{prop}
\begin{proof}
Recall that a formula in $\eqso(\LL)$ is defined by the following grammar:
\[
\alpha = \varphi \ \mid \ s \ \mid \ (\alpha + \alpha) \ \mid \ \sa{x} \alpha \ \mid \ \sa{X} \alpha
\]
where $\varphi$ is a formula in $\LL$ and $s\in\nat$. 
To find an equivalent formula in $\LL$-SNF for every $\alpha \in \eqso(\LL)$, we give a recursive function $\tau$ such that $\tau(\alpha)$ is in $\LL$-SNF and $\tau(\alpha) \equiv \alpha$. 
Specifically, if $\alpha = \varphi$, define $\tau(\alpha) = \alpha$; 
if $\alpha = s$, define $\tau(\alpha) = (\top \add \overset{\text{$s $ times}}{\ldots} \add \top)$;
if $\alpha = (\alpha_1 + \alpha_2)$, define $\tau(\alpha) = (\tau(\alpha_1) + \tau(\alpha_2))$;
if $\alpha = \sa{x}\beta$, assume $\tau(\beta) = \sum_{i = 1}^{k}\beta_i$ such that each $\beta_i$ is in $\LL$-PNF, and define $\tau(\alpha) = \sum_{i = 1}^{k}\sa{x}\beta_i$;
and if $\alpha = \sa{X}\beta$, then we proceed analogously as in the previous case.
This covers all possible cases for $\alpha$ and we conclude the proof by taking $\tau(\alpha)$ as the desired rewrite of $\alpha$.

\end{proof}
If a formula is in $\LL$-PNF then clearly the formula is in $\LL$-SNF.
Unfortunately, for some $\LL$ there exist formulae in $\eqso(\LL)$  that cannot be rewritten in $\LL$-PNF.
Therefore, to unveil the relationship between the $\sfo$-hierarchy and the $\eqso(\fo)$-hierarchy, we need to understand the boundary between PNF and SNF. We do this in the following theorem. 
\begin{thm}\label{theo-pi1-pnf}
For $i = 0,1$, there exists a formula $\alpha_i$ in $\QE{i}$ that is not equivalent to any formula in $\Sigma_i$-PNF. 
On the other hand, if $\logu{1} \subseteq \LL$ and $\LL$ is closed under conjunction and disjunction, then every formula in $\eqso(\LL)$ can be rewritten in $\LL$-PNF. 
\end{thm}
\begin{proof}

From now on, for every first-order tuple $\bar{x}$ or second-order tuple $\bar{X}$ we write $\length{\bar{x}}$ or $\length{\bar{X}}$ as the number of variables in $\bar{x}$ or $\bar{X}$ respectively. 
We divide the proof in three parts.

First, we prove that the formula $\alpha_{0} = \left( \sa{X} 1 \right) + 1$ with $\arity(X) = 1$ (i.e.\ the function $2^{\vert\A\vert}+1$) is not equivalent to any formula in $\loge{0}$-PNF. Suppose that there exists some formula $\alpha = \sa{\bar{X}}\sa{\bar{x}}\varphi(\bar{X},\bar{x})$ in $\loge{0}$-PNF that is equivalent to $\alpha_0$.
In \cite{SalujaST95}, it was proved that if $\length{\bar{X}} > 0$, the function defined by $\alpha$ is always even for big enough structures, which is not possible in our case.
On the other hand, if $\alpha$ is of the form $\sa{\bar{x}}\varphi(\bar{x})$, then $\alpha$ defines a polynomially bounded function which leads to a contradiction.

Second, we prove that the formula $\alpha_{1} = 2$ (i.e.\ $\sem{\alpha_{1}}$ is the constant function $2$) is not equivalent to any formula in $\loge{1}$-PNF. 
Suppose that there exists some formula $\alpha = \sa{\bar{X}}\sa{\bar{x}}\exists\bar{y}\, \varphi(\bar{X},\bar{x},\bar{y})$ in $\loge{1}$-PNF that is equivalent to $\alpha_1$. 
First, if $\length{\bar{X}} = \length{\bar{x}} = 0$, then the function defined by $\alpha$ is never greater than 1. 
Therefore, suppose that $\length{\bar{X}} > 0$ or $\length{\bar{x}} > 0$, and consider some ordered structure $\A$. 
Since $\sem{\alpha}(\A) = 2$, there exist at least two assignments $(\bar{B}_1,\bar{b}_1,\bar{a}_1)$, $(\bar{B}_2,\bar{b}_2,\bar{a}_2)$ to $(\bar{X},\bar{x},\bar{y})$ such that for $i\in\{1,2\}$: $\A\models\varphi(\bar{B}_i,\bar{b}_i,\bar{a}_i)$. 
Now consider the ordered structure $\A'$ that is obtained by taking the disjoint union of $\A$ twice. 
Indeed, each half of $\A'$ is isomorphic to $\A$. 
Note that $\A'\models\varphi(\bar{B}_i,\bar{b}_i,\bar{a}_i)$ for $i = 1,2$ and there exists a third assignment $(\bar{B}_1',\bar{b}_1',\bar{a}_1')$ that is isomorphic to $(\bar{B}_1,\bar{b}_1,\bar{a}_1)$, in the other half of the structure, such that $\A'\models\varphi(\bar{B}_1',\bar{b}_1',\bar{a}_1')$. 
As a result, we have that $\sem{\alpha}(\A') \geq 3$ which leads to a contradiction.

For the last part of the proof, we show that if $\LL$ contains $\logu{1}$ and is closed under conjunction and disjunction, then for every formula $\alpha$ in $\eqso(\LL)$ there exists an equivalent formula in $\LL$-PNF. 
Similarly as in the proof of Theorem~\ref{theo-pnf-snf}, we show a recursive function $\tau$ that produces such a formula. 
Assume that $\alpha = \sum_{i = 1}^n \alpha_i$ is in $\LL$-SNF where each $\alpha_i$ is in $\LL$-PNF. 
Without loss of generality, we assume that each $\alpha_i = \sa{\bar{X}}\sa{\bar{x}}\varphi_i(\bar{X},\bar{x})$ with $\length{\bar{X}} > 0$ and $\length{\bar{x}} > 0$. 
If that is not the case, we can replace each $\alpha_i$ by the equivalent formula
$$
\sa{\bar{X}} \sa{Y}\sa{\bar{x}}\sa{y}(\varphi_i(\bar{X},\bar{x})\wedge \fa{z} Y(z) \wedge \fa{z} z \leq y).
$$
Now we begin describing the function $\tau$. 
If $\alpha = \sa{\bar{X}}\sa{\bar{x}}\varphi(\bar{X},\bar{x})$, then the formula is already in $\LL$-PNF so we define $\tau(\alpha) = \alpha$. 
If $\alpha = \alpha_1 + \alpha_2$, then we assume that $\tau(\alpha_1) = \sa{\bar{X}}\sa{\bar{x}}\varphi(\bar{X},\bar{x})$ and $\tau(\alpha_2) = \sa{\bar{Y}}\sa{\bar{y}}\psi(\bar{Y},\bar{y})$. 
Our construction works by identifying a ``first'' assignment for both $(\bar{X},\bar{x})$ and $(\bar{Y},\bar{y})$ and a ``last'' assignment for both $(\bar{X},\bar{x})$ and $(\bar{Y},\bar{y})$ using the following formulae:
\begin{align*}
\gamma_{\text{first}}(\bar{X},\bar{x}) & \; = \;  \bigwedge_{i = 1}^{\length{\bar{X}}} \fa{\bar{z}}\neg X_i(\bar{z}) \wedge \fa{\bar{z}}(\bar{x}\leq\bar{z}), \\
\gamma_{\text{last}}(\bar{X},\bar{x}) & \; = \;  \bigwedge_{i = 1}^{\length{\bar{X}}} \fa{\bar{z}} X_i(\bar{z}) \wedge \fa{\bar{z}}(\bar{z}\leq\bar{x}).
\end{align*}
Similarly, we can define the formulae $\gamma_{\text{first}}(\bar{Y},\bar{y})$ and $\gamma_{\text{last}}(\bar{Y},\bar{y})$.
In other words, the ``first'' assignment is the one where every second-order predicate is empty and the first-order assignment is the lexicographically smallest, and the ``last'' assignment is the one where every second-order predicate is full and the first-order assignment is the lexicographically greatest. 
We also need to identify the assignments that are not first and the ones that are not last. 
We do this by negating the two formulae above and grouping together the first-order variables:
\begin{align*}
\gamma_{\text{not-first}}(\bar{X},\bar{x}) & \; = \; \ex{\bar{z}}(\bar{z}_0 < \bar{x} \vee \bigvee_{i = 1}^{\length{\bar{X}}}X(\bar{z}_i)), \\
\gamma_{\text{not-last}}(\bar{X},\bar{x}) & \; = \; \ex{\bar{z}}(\bar{x} < \bar{z}_0 \vee \bigvee_{i = 1}^{\length{\bar{X}}}\neg X(\bar{z}_i)),
\end{align*}
where $\bar{z} = (\bar{z}_0,\bar{z}_1,\ldots,\bar{z}_{\length{\bar{X}}})$. Then the following formula is equivalent to $\alpha$:
\begin{align}
\sa{\bar{X}}\sa{\bar{x}}\sa{\bar{Y}}\sa{\bar{y}}[&(\varphi(\bar{X},\bar{x})\wedge\gamma_{\text{not-first}}(\bar{X},\bar{x})\wedge\gamma_{\text{first}}(\bar{Y},\bar{y}))\vee \label{eq:partition1} \\
&(\varphi(\bar{X},\bar{x})\wedge\gamma_{\text{first}}(\bar{X},\bar{x})\wedge\gamma_{\text{last}}(\bar{Y},\bar{y}))\vee \label{eq:partition2}\\
&(\psi(\bar{Y},\bar{y})\wedge\gamma_{\text{first}}(\bar{X},\bar{x})\wedge\gamma_{\text{not-last}}(\bar{Y},\bar{y}))\vee \label{eq:partition3}\\
&(\psi(\bar{Y},\bar{y})\wedge\gamma_{\text{last}}(\bar{X},\bar{x})\wedge\gamma_{\text{last}}(\bar{Y},\bar{y}))]. \label{eq:partition4}
\end{align}

To show that the formula is indeed equivalent to $\alpha$, note that the formulae in lines (\ref{eq:partition1}) and (\ref{eq:partition2}) form a partition over the assignments of $(\bar{X},\bar{x})$, while fixing an assignment for $(\bar{Y},\bar{y})$, and the formulae in lines (\ref{eq:partition3}) and (\ref{eq:partition4}) form a partition over the assignments of $(\bar{Y},\bar{y})$, while fixing an assignment for $(\bar{X},\bar{x})$. 
Altogether the four lines define pairwise disjoint assignments for $(\bar{X},\bar{x}),(\bar{Y},\bar{y})$. 
With this, it is straightforward to show that the above formula is equivalent to $\alpha$. 
However, the formula is not yet in the correct form since it has existential quantifiers in the sub-formulae $\gamma_{\text{not-first}}$ and $\gamma_{\text{not-last}}$. 
To solve this, we can replace each existential quantifier by a first order sum that counts just the first assignment that satisfies the inner formula and this can be defined in $\logu{1}$. 
A similar construction was used in \cite{SalujaST95}. 

Finally, consider a $\eqso(\LL)$ formula $\alpha$ in $\LL$-SNF. 
If $\alpha = \sum_{i = 1}^n\alpha_i$, then by induction we consider $\alpha = \alpha_1 + (\sum_{i = 2}^n\alpha_i)$ and use $\tau(\alpha_1 + \tau(\sum_{i = 2}^n\alpha_i))$ as the rewrite of $\alpha$, which satisfies the hypothesis.

\end{proof}

\begin{figure*}
\begin{center}
	\begin{tikzpicture}
	\node[rectw] (n1) {$\E{0}$};
	\node[rectw, right=0.5cm of n1] (n2) {};
	\node[rectw, above=0.3cm of n2] (n3) {$\E{1}$}
		edge[draw=white] node {\rotatebox{45}{$\subsetneq$}} (n1);
	\node[rectw, below=0.5cm of n2] (n4) {$\QE{0}$}
		edge[draw=white] node {\rotatebox{315}{$\subsetneq$}} (n1);
	\node[rectw, right=0.5cm of n2] (n5) {$\QE{1}$}
		edge[draw=white] node {\rotatebox{315}{$\subsetneq$}} (n3)
		edge[draw=white] node {\rotatebox{45}{$\subsetneq$}} (n4);
	\node[rectw, right=0.8cm of n5] (n6) {$\QU{1}$}       
		edge[draw=white] node {$\subsetneq$} (n5);
	\node[rectw, below=0.3cm of n6] (n7) {$\U{1}$}       
		edge[draw=white] node {\rotatebox{90}{$=$}} (n6);
	\node[rectw, right=1.0cm of n6] (n8) {$\QE{2}$}       
		edge[draw=white] node {$\subsetneq$} (n6);
	\node[rectw, below=0.3cm of n8] (n9) {$\E{2}$}       
		edge[draw=white] node {\rotatebox{90}{$=$}} (n8);        
	\node[rectw, right=1.0cm of n8] (n10) {$\QU{2}$}       
		edge[draw=white] node {$\subsetneq$} (n8);
	\node[rectw, below=0.3cm of n10] (n11) {$\U{2}$}       
		edge[draw=white] node {\rotatebox{90}{$=$}} (n10); 
	\node[rectw, right=0.5cm of n10] (n12) {$\sfo$}       
		edge[draw=white] node {$=$} (n10); 
	\end{tikzpicture}
\end{center}
\caption{The relationship between the $\sfo$-hierarchy and the $\eqso(\fo)$-hierarchy, where $\E{1}$ and $\QE{0}$ are incomparable. \label{fig-sfo-eqso}}
\vspace{-0.1cm}
\end{figure*}

As a consequence of Proposition~\ref{theo-pnf-snf} and Theorem~\ref{theo-pi1-pnf}, we obtain that $\E{i} \subsetneq \QE{i}$ for $i = 0,1$, and that $\sh{\LL} = \eqso(\LL)$ for $\LL$ equal to  $\Pi_1$, $\Sigma_2$ or $\Pi_2$. The following proposition completes our picture of the relationship between the $\sfo$-hierarchy and the $\eqso(\fo)$-hierarchy.
\begin{prop}\label{prop-rest}
The following properties hold:
\begin{itemize}
\item $\QE{0}$ and $\E{1}$ are incomparable, that is, $\E{1} \not\subseteq \QE{0}$ and $\QE{0} \not\subseteq \E{1}$,
\item $\QE{1} \subsetneq \QU{1}$.
\end{itemize}
\end{prop}
\begin{proof}

We divide this proof into three parts.
First, we show that $\E{1} \not\subseteq \QE{0}$. 
For this inclusion to be true, it is required to hold for an arbitrary ordered relational signature $\R$, so it suffices to show that it  is not true for at least one such a signature.
Let $\R$ be the ordered signature that contains only the relation name $<$.
Suppose that there is a $\QE{0}$ formula $\alpha$ over $\R$ which is equivalent to the $\E{1}$-formula $\sa{x} \ex{y} (x < y)$. 
That is, for every finite $\R$-structure $\A$, $\sem{\alpha}(\A) = \size{\A} - 1$.

Suppose that $\alpha$ is in SNF, namely, $\alpha = \sum_{i = 1}^k \alpha_i$ for some fixed $k$. 
Since $\sem{\alpha}$ is not the identically zero function, consider some $\alpha_i$ that describes a non-null function. 
Let $\alpha_i = \sa{\bar{X}}\sa{\bar{x}}\varphi(\bar{X},\bar{x})$ where $\varphi$ is quantifier-free. 
Notice that if $\length{\bar{X}} > 0$, then the function $\sem{\alpha}$ is in $\Omega(2^{\size{\A}})$, as it was proven in~\cite{SalujaST95}. 
Therefore, we have that $\alpha_i = \sa{\bar{x}}\varphi(\bar{x})$. 
We conclude our proof with the following claim, from which we conclude that $\alpha = \sum_{i = 1}^k \alpha_i$ cannot be equivalent to the $\E{1}$-formula $\sa{x} \ex{y} (x < y)$. 
\begin{clm}
	Let $\beta = \sa{\bar{x}}\varphi(\bar{x})$	where $\varphi$ is quantifier free. 
	Then the function $\sem{\beta}$ is either null, greater or equal to $n$, or is in $\Omega(n^2)$, where $n$ is the size of the input structure.
\end{clm}
\begin{proof}
Assume that $\bar x = (x_1, \ldots, x_m)$, and notice that each atomic sub-formula in $\varphi(\bar{x})$ is either $(x_i = x_j)$, $(x_i < x_j)$, $\top$ or a negation thereof, where $i,j \in \{1, \ldots, m\}$. 
	Suppose $\sem{\beta}$ is not null and consider some $\R$-structure $\A$ such that $\sem{\beta}(\A) > 0$. Hence, there exists an assignment $\bar a = (a_1, \ldots, a_m)$
	for $\bar{x}$ such that $\A\models\varphi(\bar{a})$.
	Given this assignment, define an equivalence relation $\sim$ on $\{x_1, \ldots, x_m\}$ as follows: $x_i \sim x_j$ if and only if $a_i = a_j$, and assume that $\sim$ partitions $\{x_1, \ldots, x_m\}$ into $\ell$ equivalence classes, where $\ell \geq 1$.
Then we have that there exist at least $\binom{\size{\A}}{\ell}$ assignments $\bar b$ for $\bar{x}$ such that $\A\models\varphi(\bar{b})$. Thus, given that if $\ell = 1$, then $\binom{\size{\A}}{\ell} = \size{\A}$, and if $\ell \geq 2$, then $\binom{\size{\A}}{\ell} \in \Omega(\size{\A}^2)$, we conclude that the claim holds.
\end{proof}

Now we show that $\QE{0} \not\subseteq \E{1}$. 
In Theorem \ref{theo-pi1-pnf} we proved that there is no formula in $\loge{1}$-PNF equivalent to the formula $\alpha = 2$. 
Every formula in $\E{1}$ can be expressed in $\loge{1}$-PNF, which implies that $2 \not\in \E{1}$. Therefore, given that $2 \in \QE{0}$ by the definition of this logic, we conclude that $\QE{0} \not\subseteq \E{1}$.

Finally, we prove that $\eqso(\loge{1})\subsetneq\eqso(\logu{1})$. 
For inclusion, let $\alpha$ be a formula in $\eqso(\loge{1})$. 
Suppose that it is in $\loge{1}$-SNF, namely, $\alpha = c + \sum_{i = 1}^{n}\alpha_i$. 
Let $\alpha_i = \sa{\bar{X}}\sa{\bar{x}}\ex{\bar{y}}\varphi_i(\bar{X},\bar{x},\bar{y})$, where $\varphi_i$ is quantifier-free for each $\alpha_i$. 
We use the same construction used in \cite{SalujaST95}, and we obtain that the formula $\ex{\bar{y}}\varphi_i(\bar{X},\bar{x},\bar{y})$ is equivalent to $\sa{\bar{y}}\,[\varphi_i(\bar{X},\bar{x},\bar{y}) \wedge \fa{\bar{y}'}(\varphi_i(\bar{X},\bar{x},\bar{y}')\to\bar{y}\leq\bar{y}')]$ for every assignment to $(\bar{X},\bar{x})$. 
We do this replacement for each $\alpha_i$, and we obtain an equivalent formula to $\alpha$ in $\eqso(\logu{1})$.

To prove that the inclusion is proper, consider the $\eqso(\logu{1})$ formula $\sa{x} \fa{y}(y = x)$. 
This formula defines the following function over each ordered structure $\A$:
$$
\sem{\alpha}(\A) = 
\begin{cases}
1 &\A \text{ has one element}\\
0 &\text{ otherwise}.
\end{cases}
$$
Suppose that there exists an equivalent formula $\alpha$ in $\eqso(\loge{1})$. 
Also, suppose that it is in $\LL$-SNF, so $\alpha = \sum_{i = 1}^n\sa{\bar{X}}\sa{\bar{x}}\ex{\bar{y}}\varphi_i(\bar{X},\bar{x},\bar{y})$. 
Consider a structure $\A'$ with one element. 
We have that for some $i$, there exists an assignment $(\bar{B},\bar{b},\bar{a})$ for $(\bar{X},\bar{x},\bar{y})$ such that $\A' \models\varphi_i(\bar{B},\bar{b},\bar{a})$. 
Consider now the structure $\A''$ that is obtained by duplicating $\A'$, as we did for Theorem \ref{theo-pi1-pnf}. 
Note that $\A''\models\varphi_i(\bar{B},\bar{b},\bar{a})$, which implies that $\sem{\alpha}(\A' \uplus \A'') > 1$, which leads to a contradiction.
%

\end{proof}
The relationship between the two hierarchies is summarized in Figure \ref{fig-sfo-eqso}.
Our hierarchy and the one proposed in~\cite{SalujaST95} only differ in~$\Sigma_0$ and~$\Sigma_1$. 
Interestingly, we show next that this difference is crucial for finding classes of functions with easy decision versions and good closure properties.


\subsection{Defining a class of functions with easy decision versions and good closure properties}
\label{sec-clo}

We use the \emph{$\eqso(\fo)$-hierarchy} to define syntactic classes of functions with good algorithmic and closure properties.  
But before doing this, we introduce a more strict notion of counting problem with easy decision version.
Recall that a function $f\colon\Sigma^* \to \mathbb{N}$ has an easy decision counterpart if $L_f = \{ x \in \Sigma^* \mid f(x) > 0 \}$ is a language in $\ptime$. As the goal of this section is to define a syntactic class of functions in $\shp$ with easy decision versions and good closure properties, we do not directly consider the semantic condition $L_f  \in \ptime$, but instead we consider a more restricted syntactic condition. More precisely, a function $f\colon\Sigma^* \to \mathbb{N}$ is said to be in the complexity class $\totp$~\cite{PagourtzisZ06} if there exists a  polynomial-time NTM $M$ such that $f(x) = \tmt_M(x) - 1$ for every $x \in \Sigma^*$, where $\tmt_M(x)$ is the total number of runs of $M$ with input $x$. Notice that one is subtracted from $\tmt_M(x)$ to allow for $f(x) = 0$. Besides, notice that $\totp \subseteq \shp$ and that $f \in \totp$ implies that $L_f \in \ptime$. 

The complexity class $\totp$ contains many important counting problems with easy decision counterparts, such as $\cpm$, $\cdnf$, and $\chsat$ among others~\cite{PagourtzisZ06}. Besides, $\totp$ has good closure properties as it is closed under sum, multiplication and subtraction by one. However, some functions in $\totp$ do not admit FPRAS under standard complexity-theoretical assumptions.\footnote{As an example consider the problem of counting the number of independent sets in a graph, and the widely believed assumption that $\np$ is not equal to the randomized complexity class $\rp$ (Randomized Polynomial-Time \cite{G77}). This counting problem is in $\totp$, and it is known that $\np = \rp$ if there exists an FPRAS for it \cite{DFJ02}.}
Hence, we use the $\eqso(\fo)$-hierarchy to find restrictions of $\totp$ with good approximation and closure properties.

It was proved in \cite{SalujaST95} that every function in $\E{1}$ admits an FPRAS. Besides, it can be proved that $\E{1} \subseteq \totp$. 
However, this class is not closed under sum, so it is not robust under the basic closure properties we are looking for. 
\begin{prop}\label{prop-e1-nc}
There exist functions $f, g \in \E{1}$ such that $(f + g) \not\in \E{1}$.
\end{prop}
\begin{proof}
Towards a contradiction, assume that $\E{1}$ is closed under binary sum. 
Consider the formula $\alpha = \sa{x}(x = x) \in \E{1}$ over some signature $\R$. 
This defines the function $\sem{\alpha}(\A) = \size{\A}$. 
From our assumption, there exists some formula in $\E{1}$ equivalent to the formula $\alpha \add \alpha$, which describes the function $2\size{\A}$. 
Let $\sa{\bar{X}}\sa{\bar{x}}\exists\bar{y}\,\varphi(\bar{X},\bar{x},\bar{y})$ be this formula, where $\varphi$ is first-order and quantifier-free. 
For each $\R$-structure $\A$, we have the following inequality:
$$
\sem{\sa{\bar{X}}\sa{\bar{x}}\sa{\bar{y}}\,\varphi(\bar{X},\bar{x},\bar{y})}(\A)
\leq 
\sem{\sa{\bar{X}}\sa{\bar{x}}\exists\bar{y}\,\varphi(\bar{X},\bar{x},\bar{y})}(\A) \cdot  \size{\A}^{\length{\bar{y}}} \leq 2\size{\A}^{\length{\bar{y}}+1} 
$$
Note that the formula $\sa{\bar{X}}\sa{\bar{x}}\sa{\bar{y}}\,\varphi(\bar{X},\bar{x},\bar{y})$ defines a function in $\E{0}$. 
Therefore, as it was proven in \cite{SalujaST95}, if $\length{\bar{X}} > 0$ then the function is in $\Omega(2^{\size{\A}})$, which violates the inequality for large structures.

We now have that $\length{\bar{X}} = 0$.
Consider a structure $\mathfrak{1}$ with only one element. 
We have that $\sem{\sa{\bar{x}}\exists\bar{y}\,\varphi(\bar{x},\bar{y})}(\mathfrak{1}) = 2$, but since the structure has only one element, there is only one possible assignment to $\bar{x}$. 
And so, $\sem{\sa{\bar{x}}\exists\bar{y}\,\varphi(\bar{x},\bar{y})}(\mathfrak{1}) \leq 1$, which leads to a contradiction.
%

\end{proof}
To overcome this limitation, one can consider the class $\QE{1}$, which is closed under sum by definition. In fact, the following proposition shows that the same good properties as for $\E{1}$ hold for $\QE{1}$, together with the fact that it is closed under sum and multiplication.
\begin{prop} \label{prop:qe0-fp-qe1-totp-fptras}
$\QE{1} \subseteq \totp$ and every function in $\QE{1}$ has an FPRAS. Moreover, $\QE{1}$ is closed under sum and multiplication.
\end{prop}
\begin{proof}

The authors in \cite{SalujaST95} proved that there exists a {\em product reduction} from every function in $\E{1}$ to a restricted version of $\cdnf$. 
That is, if $\alpha\in\E{1}$ over some signature $\R$, there exist polynomially computable functions $g\colon\ostr[\R]\to\ostr[\R_{\text{DNF}}]$ and $h\colon\nat\to\nat$ such that for every $\R$-structure $\A$, it holds that $\sem{\alpha}(\A) = \cdnf(\enc(g(\A)))\cdot h(\size{\A})$. We use this fact in the following arguments. 

To show that $\eqso(\loge{1})$ is contained  in \textsc{TotP}, let $\alpha$ be a $\eqso(\loge{1})$ formula and assume that it is in $\loge{1}$-SNF. 
That is, $\alpha = \sum_{i = 1}^n\alpha_i$ where each $\alpha_i$ is in $\loge{1}$-PNF. 
Consider the following nondeterministic procedure that on input $\enc(\A)$ generates $\sem{\alpha}(\A)$ branches. 
For each $\alpha_i = \varphi$, where $\varphi$ is a $\loge{1}$ formula, it checks if $\A\models\varphi$ in polynomial time and generates a new branch if that is the case. 
For each $\alpha_i = \sa{\bar{X}}\sa{\bar{x}}\varphi$, this formula is also in $\E{1}$. 
We use the reduction to $\cdnf$ provided in \cite{SalujaST95} and we obtain $g(\enc(\A))$, which is an instance to $\cdnf$. 
Since $\cdnf$ is also in $\totp$ \cite{PagourtzisZ06}, we simulate the corresponding nondeterministic procedure that generates exactly $\cdnf(\enc(g(\A)))$ branches. 
Since, $\fp\subseteq\totp$\cite{PagourtzisZ06}, each polynomially computable function is also in $\totp$, and then on each of these branches we simulate the corresponding nondeterministic procedure to generate $h(\size{\A})$ more. 
The number of branches for each $\alpha_i$ is $\sem{\alpha_i}(\A) = \cdnf(\enc(g(\A)))\cdot h(\size{\A})$, and the total number of branches is equal to $\sem{\alpha}(\A)$. 
We conclude that $\alpha\in\totp$.

To show that every function in $\eqso(\loge{1})$ has an FPRAS,  let $\alpha$ be a $\eqso(\loge{1})$ formula and assume that it is in $\loge{1}$-SNF. 
That is, $\alpha = \sum_{i = 1}^n\alpha_i$ where each $\alpha_i$ is in $\loge{1}$-PNF. 
Note that each $\alpha_i$ that is equal to some $\loge{1}$ formula $\varphi$ has an FPRAS given by the procedure that simply checks if $\A\models\varphi$ and returns 1 if it does and 0 otherwise. 
Also, each remaining $\alpha_i$ has an FPRAS since $\alpha_i\in \E{1}$ \cite{SalujaST95}. 
If two functions have an FPRAS, then their sum also has one given by the procedure that simulates them both and sums the results. 
We conclude that $\alpha$ has an FPRAS.

Finally, we show that $\eqso(\loge{1})$ is closed under sum and multiplication. 
Since $\eqso(\loge{1})$ is closed under sum by definition, we focus only on proving that it is closed under multiplication. 
We prove this for the more general case of $\eqso(\LL)$ with $\LL$ being a fragment of $\so$.

\begin{lem} \label{conj-mult}
	If $\LL$ is a fragment closed under conjunction, then $\eqso(\LL)$ is closed under binary multiplication.
\end{lem}
\begin{proof}
	Given two formulae $\alpha, \beta$ in $\eqso(\LL)$ we will construct a formula in the logic which is equivalent to $(\alpha\mult \beta)$. 
	From what was proven in Proposition \ref{theo-pnf-snf}, we may assume that $\alpha$ and $\beta$ are in $\LL$-SNF. 
	Let $\alpha = \sum_{i = 1}^n\sa{\bar{X}_i}\sa{\bar{x}_i}\varphi_i(\bar{X}_i,\bar{x}_i)$, and $\beta = \sum_{i = 1}^m\sa{\bar{Y}_i}\sa{\bar{y}_i}\psi_i(\bar{Y}_i,\bar{y}_i)$. Expanding the product in $(\alpha\mult \beta)$ and reorganizing results in the equivalent formula
	$$
	\sum_{i = 1}^n\sum_{j = 1}^m\sa{\bar{X}_i}\sa{\bar{Y}_j}\sa{\bar{x}_i}\sa{\bar{y}_j}(\varphi_i(\bar{X}_i,\bar{x}_i)\wedge\psi_i(\bar{Y}_j,\bar{y}_j)),
	$$
	which is in $\LL$-SNF, and therefore, in $\eqso(\LL)$.
\end{proof}
Since $\loge{1}$ is closed under conjunction, we have that Lemma~\ref{conj-mult} can be applied to $\eqso(\loge{1})$, and we can deduce that $\eqso(\loge{1})$ is closed under multiplication. 
This concludes the proof of the proposition. 
%

\end{proof}
Hence, it only remains to prove that $\QE{1}$ is closed under subtraction by one. Unfortunately, it is not clear whether this property holds; in fact, we conjecture that it is not the case. Thus, we need to find an extension of $\QE{1}$ that keeps all the previous properties and is closed under subtraction by one. It is important to notice that $\shp$ is believed not to be closed under subtraction by one due to some complexity-theoretical assumption\footnote{A decision problem $L$ is in the randomized complexity class $\cspp$ if there exists a polynomial-time NTM $M$ such that for every $x \in L$ it holds that $\tma_M(x) - \tmr_M(x) = 2$, and for every $x \not\in L$ it holds that $\tma_M(x) = \tmr_M(x)$ \cite{OH93,FFK94}. It is believed that $\np \not\subseteq \cspp$.
However, if $\shp$ is closed under subtraction by one, then it holds that $\np \subseteq \cspp$ \cite{OH93}.}. So, the following proposition rules out any logic that extends $\Pi_1$ as a possible extension of~$\QE{1}$ with the desired closure property.
\begin{prop} \label{pi-minusone}
If $\Pi_1 \subseteq \LL \subseteq \fo$ and $\eqso(\LL)$ is closed under subtraction by one, then $\shp$ is closed under subtraction by~one. 
\end{prop}
\begin{proof}
Let $\LL$ be a fragment of $\fo$ that contains $\logu{1}$. Then we have that every function in $\U{1}$ is expressible in $\eqso(\LL)$. In particular, $\ctcnf \in \eqso(\LL)$. Suppose that $\eqso(\LL)$ is closed under subtraction by one. Then, the function $\ctcnf-1$, which counts the number of satisfying assignments of a 3-CNF formula minus one, is also in $\eqso(\LL)$. Recall also that $\eqso(\LL) \subseteq \eqso(\fo) = \shp$ and that $\ctcnf$ is $\shp$-complete under parsimonious reductions\footnote{It can be easily verified that the standard reduction from SAT to 3-CNF (or 3-SAT) preserves the number of satisfying assignments}. Let $f$ be a function in $\shp$, and consider the nondeterministic polynomial-time procedure that on input $\enc(\A)$ computes the corresponding reduction $g(\enc(\A))$ into $\ctcnf$ and simulates the $\shp$ procedure for $\ctcnf-1$ on input $g(\enc(\A))$. This is a $\shp$ procedure that computes $f-1$, from which we conclude that $\shp$ is closed under subtraction by one.

\end{proof}
Therefore, the desired extension has to be achieved by allowing some local extensions to~$\Sigma_1$. To this end, we define $\logex{1}$ as $\Sigma_1$ but allowing atomic formulae over a signature~$\R$ to be of the form either $u = v$ or $X(\bar u)$, where $X$ is a second-order variable, or $\varphi(\bar u)$, where $\varphi(\bar u)$ is a first-order formula over $\R$ (in particular, it does not mention any second-order variable). With this extension we obtain a class with the desired properties.
\begin{thm}\label{sigmafo-minusone}
The class $\eqso(\logex{1})$ is closed under sum, multiplication and subtraction by one. Moreover, $\eqso(\logex{1}) \subseteq \totp$ and every function in $\eqso(\logex{1})$ has an FPRAS.
\end{thm}
\begin{leftbar}
  \begin{proof}

For the sake of readability, we divide the proof into three parts. The last part, i.e.\ subtraction by one, is the most technical proof in the paper. Since this proof is several pages long, a line was drawn on the left margin to visually differentiate it from the rest of the paper.
\medskip

\noindent {\bf Closed under sum and multiplication.} By the previous results, it is straightforward to prove that $\eqso(\logex{1})$ is closed under sum and multiplication. Indeed, $\eqso(\LL)$ is closed under sum by definition for every fragment $\LL$, and since $\logex{1}$ is closed under conjunction, from Lemma \ref{conj-mult} it follows that $\eqso(\logex{1})$ is closed under multiplication.
\medskip

\noindent {\bf Easy decision version and FPRAS.}  We show here that $\eqso(\logex{1}) \subseteq \totp$ and every function in $\eqso(\logex{1})$ has an FPRAS. We do this by showing a parsimonious reduction from any function in $\eqso(\logex{1})$ to some function in $\eqso(\loge{1})$, and using the result of Proposition~\ref{prop:qe0-fp-qe1-totp-fptras}. First, we define a function that converts any formula $\alpha$ in $\eqso(\logex{1})$ over a signature $\R$ into a formula $\lambda(\alpha)$ in $\eqso(\loge{1})$ over a signature $\R_{\alpha}$. Afterwards, we define a function $g_{\alpha}$ that receives an $\R$-structure $\A$ and outputs an $\R_{\alpha}$-structure $g_{\alpha}(\A)$.

Let $\alpha$ be in $\eqso(\logex{1})$. The signature $\R_{\alpha}$ is obtained by adding the symbol $R_{\psi}$ to $\R$ for every $\fo$ formula $\psi(\bar{z})$ in $\alpha$. Each symbol $R_{\psi}$ represents a predicate with arity $\length{\bar{z}}$. Then, $\lambda(\alpha)$ is defined as $\alpha$ where each $\fo$ formula $\psi(\bar{z})$ has been replaced by $R_{\psi}(\bar{z})$. We now define the function $g_{\alpha}$ with a polynomial time procedure. Let $\A$ be a $\R$-structure with domain $A$. Let $\A'$ be an $\R_{\alpha}$-structure obtained by copying $\A$ and leaving each $R_{\psi}^{\A'}$ empty. For each $\fo$-formula $\psi(\z)$ with $\length{\bar{z}}$ open first-order variables, we iterate over all tuples $\bar{a} \in A^{\length{\bar{z}}}$. If $\A\models\psi(\bar{a})$ (this can be done in $\ptime$), then the tuple $\bar{a}$ is added to $R_{\psi}^{\A'}$ . This concludes the construction of $\A'$. Note that the number of $\fo$ subformulae, arity and tuple size is fixed in $\alpha$, so computing this function takes polynomial time over the size of the structure. Moreover, the encoding of $\A'$ has polynomial size over the size of $\enc(\A)$. We define $g_{\alpha}(\A) = \A'$ and we have that for each $\R$-structure $\A$: $\sem{\alpha}(\A) = \sem{\lambda(\alpha)}(g_{\alpha}(\A))$. Therefore, we have a parsimonious reduction from $\alpha$ to the $\eqso(\loge{1})$ formula $\lambda(\alpha)$.

To show that the function defined by $\alpha$ is in $\totp$, we can convert $\alpha$ and $\enc(\A)$ into $\lambda(\alpha)$ and $\enc(g_{\alpha}(\A))$, respectively, and run the procedure in Proposition~\ref{prop:qe0-fp-qe1-totp-fptras}. Similarly, to show that $\alpha$ has an FPRAS, we do the same as before and simulates the FPRAS for $\lambda(\alpha)$ in Proposition~\ref{prop:qe0-fp-qe1-totp-fptras}. These procedures also takes polynomial time and satisfies the required conditions.

\medskip

\noindent {\bf Closed under subtraction by one.} We prove here that $\eqso(\logex{1})$ is closed under subtraction by one. 
For this, given $\alpha \in \eqso(\logex{1})$ over a signature $\R$, we will define a $\eqso(\logex{1})$-formula $\kappa(\alpha)$ such that for each finite structure $\A$ over $\R$: $\sem{\kappa(\alpha)}(\A) = \sem{\alpha}(\A) \dotminus 1$. 
Without loss of generality, we assume that $\alpha$ is in $\logex{1}$-SNF, that is, $\alpha = \sum_{i = 1}^{n}\sa{\bar{X}}\sa{\bar{x}}\varphi_i$ where each $\varphi_i$ is in $\logex{1}$. Moreover, we assume that $\length{\bar{x}} > 0$  since, if this is not the case, we can replace $\sa{\bar{X}}\varphi_i$ with the equivalent formula $\sa{\bar{X}} \sa{y} \varphi_i \wedge\first(y)$.

{\em Proof outline:} The proof will be separated in two parts. In the first part, we will assume that  $\alpha$ is in $\logex{1}$-PNF, namely, $\alpha = \sa{\bar{X}}\sa{\bar{x}}\varphi$ for some $\varphi$ in $\logex{1}$. We will show how to define a formula $\varphi'$ that satisfies the following condition: for each $\A$, if $(\A,V,v)\models \varphi(\bar{X},\bar{x})$ for some $V$ and $v$ over $\A$, then there exists exactly one assignment to $(\bar{X},\bar{x})$ that satisfies $\varphi$ and not $\varphi'$. 
From this, we will have that $\kappa(\alpha) = \sa{\bar{X}}\sa{\bar{x}}\varphi'$ is the desired formula. In the second part, we suppose that $\alpha$ is of the form $\beta + \sa{\bar{X}}\sa{\bar{x}}\varphi$ with $\beta$ being the sum of one or more formulae in $\logex{1}$-PNF. 
We define a formula $\varphi'$ that satisfies the following condition: if $(\A,V,v)\models \varphi(\bar{X},\bar{x})$ and $\sem{\beta}(\A) = 0$, then there exists exactly one assignment to $(\bar{X},\bar{x})$ that satisfies $\varphi$ and not $\varphi'$. From here, we can define $\kappa$ recursively as  $\kappa(\alpha) = \kappa(\beta) +  \sa{\bar{X}}\sa{\bar{x}}\varphi'$ and the property of subtraction by one will be~proven.

\medskip

\noindent {\em Part (1).}  Let $\alpha =  \sa{\bar{X}}\sa{\bar{x}} \varphi(\bar{X},\bar{x})$ where $\varphi$ is a $\logex{1}$-formula.
Note that, if $\alpha$ is of the form $\alpha = \sa{\bar{x}} \varphi(\bar{x})$ (i.e.\ $\length{\bar{X}} = 0$), we can define $\kappa(\alpha) = \sa{\bar{x}} [\varphi(\bar{x})\wedge \ex{\bar{z}}(\varphi(\bar{z})\wedge\bar{z} < \bar{x})]$, which is in $\eqso(\logex{1})$ and fulfils the desired property. 
Therefore, for the rest of the proof we can assume that $\length{\bar{X}} > 0$ and $\length{\bar{x}} > 0$.

To simplify the analysis of $\varphi$, the first step is to rewrite $\varphi$ as a DNF formula. 
More precisely, we rewrite $\varphi$ into an equivalent formula of the form $\bigvee_{i = 1}^m \varphi_i$ for some $m\in \bbN$ where each $\varphi_i(\bar{X}, \bar{x}) = \ex{\bar{y}} \varphi_i'(\bar{X}, \bar{x}, \bar{y})$ and $\varphi_i'(\bar{X}, \bar{x}, \bar{y})$ is a conjunction of atomic formulae or negation of atomic formulae. Furthermore, we assume that each $\varphi_i'(\bar{X}, \bar{x}, \bar{y})$ has the~form:
$$
\varphi_i'(\bar{X}, \bar{x}, \bar{y}) =  \underbrace{\varphi_i^{\fo}(\bar{x},\bar{y})}_{\text{an $\fo$ formula}} \wedge 
\underbrace{\varphi_i^{+}(\bar{X},\bar{x},\bar{y})}_{\text{conjunction of $X_j$'s}} \wedge
\underbrace{\varphi_i^{-}(\bar{X},\bar{x},\bar{y})}_{\text{conjunction of $\neg X_j$'s}}.
$$
Note that atomic formulae, like $R(\bar{z})$ for $R \in \R$, will appear in the subformula $\varphi_i^{\fo}(\bar{x},\bar{y})$. 

Now, we define a series of rewrites of $\varphi$ that will make each formula $\varphi_i$ satisfy the following three conditions: (a) no variable from $\bar{x}$ appears in $\varphi_i^{-}(\bar{X},\bar{x},\bar{y})\wedge\varphi_i^{+}(\bar{X},\bar{x},\bar{y})$, (b) $\varphi_i^{\fo}(\bar{x},\bar{y})$ defines a weak ordering over the variables in $(\bar{x},\bar{y})$ (the precise definition of a weak ordering is detailed below) and (c) if both $X_j(\bar{z})$ and $\neg X_j(\bar{w})$ appear in the formula, the weak ordering should not satisfy $\bar{z} = \bar{w}$.
When these conditions are met, $\varphi_i'(\bar{X}, \bar{x}, \bar{y})$ will be satisfiable if and only if $\varphi_i^{\fo}(\bar{x},\bar{y})$ is satisfiable. This is proven in Claim~\ref{claim:minusone}.
We explain below how to rewrite $\varphi_i$ in order to satisfy each condition.

\medskip

\noindent {\em (a) No variable from $\bar{x}$ appears in $\varphi_i^{-}(\bar{X},\bar{x},\bar{y})\wedge\varphi_i^{+}(\bar{X},\bar{x},\bar{y})$.} In order to satisfy this condition, consider some instance of $X_j(\bar{w})$ in $\varphi_i$, where $\bar{w}$ is a subtuple of $(\bar{x},\bar{y})$. Add $\length{\bar{w}}$ new variables $z_1,\ldots,z_{\length{\bar{w}}}$ to the formula and let $\bar{z} = (z_1,\ldots,z_{\length{\bar{w}}})$. We rewrite $\varphi_i^{+}(\bar{X},\bar{x},\bar{y})$ by replacing $X_j(\bar{w})$ with $X_j(\bar{z})$ (denoted by $\varphi_i^{+}(\bar{X},\bar{x},\bar{y})[X_j(\bar{w}) \leftarrow X_j(\bar{z})]$) and then the formula $\varphi_i$ is equivalently defined as:
$$
\varphi_i(\bar{X},\bar{x}) \ := \ \ex{\bar{y}} \ex{\bar{z}} \big( \bar{z} = \bar{w} \wedge \varphi_i^{\fo}(\bar{x},\bar{y}) \wedge \varphi_i^{+}(\bar{X},\bar{x},\bar{y})[X_j(\bar{w}) \leftarrow X_j(\bar{z})] \wedge
\varphi_i^{-}(\bar{X},\bar{x},\bar{y}) \big).
$$
We repeat this process for each instance of a $X_j(\bar{w})$ in $\varphi_i$, and we obtain a formula where none of the $X_j$'s acts over any variable in $\bar{x}$. We add the new first-order variables to $\bar{y}$ and we redefine $\varphi_i$ as:
$$
\varphi_i(\bar{X},\bar{x}) \ := \  \ex{\bar{y}} \big( \varphi_i^{\fo}(\bar{x},\bar{y}) \wedge \varphi_i^{-}(\bar{X},\bar{y}) \wedge \varphi_i^{+}(\bar{X},\bar{y})\big).
$$
For example, if $\bar{x} = x$, $\bar{y} = y$ and $\varphi_i = \ex{\bar{y}} \big( x < y \wedge  X(x,y)\wedge \neg X(x,x)\big)$, then we redefine $\bar{y} = (y,v_1,v_2,v_3,v_4)$ and:
$$
\varphi_i \ := \ \ex{\bar{y}}\big( v_1 = x \wedge v_2 = y \wedge v_3 = x \wedge v_4 = x \wedge x < y  \wedge  X(v_1,v_2) \wedge \neg X(v_3,v_4)\big).
$$
\noindent {\em (b) $\varphi_i^{\fo}(\bar{x},\bar{y})$ defines a weak ordering over the variables in $(\bar{x},\bar{y})$.} A weak ordering on a set $S$ is defined by an equivalence relation $\sim$ over $S$, and a linear order over $S/\!\sim$. For example, let $\bar{x} = (x_1,x_2,x_3,x_4)$. A possible weak ordering would be defined by the formula $\theta(\bar{x}) = x_2 < x_1 \wedge x_1 = x_4 \wedge x_4 < x_3$. On the other hand, the formula $\theta'(\bar{x}) = x_1 < x_2 \wedge x_1 < x_4 \wedge x_2 = x_3$ does not define a weak ordering since both $\{x_1\}<\{x_2,x_3\}<\{x_4\}$ and $\{x_1\} < \{x_2,x_3,x_4\}$ satisfy $\theta'$.
For a given $k$, let $\cB_k$ be the number of possible weak orderings for a set of size $k$. For $1 \leq j \leq \cB_{\length{(\bar{x},\bar{y})}}$ 
let $\theta^j(\bar{x},\bar{y})$ be the formula that defines the $j$-th weak ordering over $(\bar{x},\bar{y})$. 
The formula $\varphi(\bar{X},\bar{x})$ is thus redefined as:
$$
\varphi(\bar{X},\bar{x}) \ := \ \bigvee_{i = 1}^m \bigvee_{j = 1}^{\cB_{\length{(\bar{x},\bar{y})}}} \ex{\bar{y}} \big(\theta^j(\bar{x},\bar{y})\wedge \varphi_i^{\fo}(\bar{x},\bar{y}) \wedge \varphi_i^{-}(\bar{X},\bar{y}) \wedge \varphi_i^{+}(\bar{X},\bar{y})\big),
$$
Note that each $\theta^j(\bar{x},\bar{y})$ is an $\fo$-formula.
Then, by redefining $\varphi_i^{\fo}(\bar{x},\bar{y})$  as $\theta^j(\bar{x},\bar{y})\wedge \varphi_i^{\fo}(\bar{x},\bar{y})$, we can suppose that each $\varphi_i^{\fo}(\bar{x},\bar{y})$ forces a weak ordering over the variables in $(\bar{x},\bar{y})$.

\noindent {\em (c) If both $X_j(\bar{z})$ and $\neg X_j(\bar{w})$ appear in the formula, the weak ordering should not satisfy $\bar{z} = \bar{w}$.} 
If there exists an instance of $X_j(\bar{z})$ in $\varphi^{+}_i$, an instance of $\neg X_j(\bar{w})$ in $\varphi^{-}_i$ and the weak ordering in $\varphi^{\fo}_i$ satisfies $\bar{z} = \bar{w}$, then the entire formula $\varphi_i$ is removed from $\varphi$.
It is important to notice that the resulting $\varphi$ is equivalent to the initial one, and it is still a formula in $\eqso(\logex{1})$. From now on, we assume that each $\varphi_i(\bar{X},\bar{x}) = \ex{\bar{y}}  \varphi_i'(\bar{X},\bar{x}, \bar{y}) $ satisfies conditions (a), (b) and (c), and moreover, $\varphi_i'(\bar{X},\bar{x}, \bar{y})$ has the~form:
$$
\varphi_i'(\bar{X},\bar{x}, \bar{y}) \; = \; \varphi^{\fo}_i(\bar{x},\bar{y}) \wedge \varphi^{+}_i(\bar{X}, \bar{y}) \wedge \varphi^{-}_i(\bar{X},\bar{y}).
$$
Note that neither one of $\varphi^{+}$ or $\varphi^{-}_i$ depends on $\bar{x}$.
\begin{clm}\label{claim:minusone}
	For an ordered structure $\A$ and an FO assignment $v$ for $\A$, $(\A,v)\models\varphi^{\fo}_i(\bar{x},\bar{y})$ if, and only if, there exists an SO assignment $V$ for $\A$ such that $(\A,V,v)\models \varphi_i'(\bar{X},\bar{x}, \bar{y})$.
\end{clm}
\begin{proof}
	To prove the {\em only if} direction, let $\A$ be an ordered structure with domain $A$ and let $v$ be a first-order assignment for $\A$, such that $(\A,v)\models\varphi^{\fo}_i(\bar{x},\bar{y})$.
	Define $\bar{B} = (B_1,\ldots,B_{\length{\bar{X}}})$ as $B_j = \{v(\bar{w})\mid \text{ $X_j(\bar{w})$ is mentioned in $\varphi^{+}_i(\bar{X},\bar{y})$}\}$, and let $V$ be a second-order assignment for which $V(\bar{X}) = \bar{B}$.
	Towards a contradiction, suppose that $(\A,V,v)\not\models \varphi^{\fo}_i(\bar{x},\bar{y}) \wedge \varphi^{+}_i(\bar{X}, \bar{y}) \wedge \varphi^{-}_i(\bar{X},\bar{y})$.
	By the choice of $v$, and construction of $V$ it is clear that $(\A,V,v)\models\varphi^{\fo}_i(\bar{x},\bar{y})\wedge\varphi^{+}_i(\bar{X},\bar{y})$, so we necessarily have that $(\A,V,v)\not\models\varphi^{-}_i(\bar{X},\bar{y})$.
	Let $X_t$ be such that $\neg X_t(\bar{w})$ is mentioned in $\varphi^{-}_i(\bar{X},\bar{y})$ and $(\A,V, v)\not\models\neg X_t(\bar{w})$, namely, $v(\bar{w})\in B_t$. 
	However, by the construction of $B_t$, there exists a subtuple $\bar{z}$ of $\bar{y}$ such that $X_t(\bar{z})$ appears in $\varphi^{+}_i(\bar{X},\bar{y})$ and $v(\bar{z}) = v(\bar{w})$. Since $(\A,v)\models\varphi^{\fo}_i(\bar{x},\bar{y})$ and $v(\bar{z}) = v(\bar{w})$, then the weak ordering in $\varphi^{\fo}_i$ satisfies $\bar{z} = \bar{w}$. This violates condition (c) above since $\neg X_t(\bar{w})$ appears in $\varphi^{-}_i$ and $X_t(\bar{z})$ appears in $\varphi^{+}_i$, which leads to a contradiction. 
	
The other direction is trivial since $\varphi^{\fo}_i(\bar{x},\bar{y})$ is a subformula of $\varphi_i'$.
\end{proof}

The previous claim its subsequent proof motivate the following definition. For a first-order assignment $v$ over $\A$, define $\bar{B}^v = (B^v_1,\ldots,B^v_{\length{\bar{X}}})$ where each $B^v_j = \{v(\bar{w}) \mid \text{$X_j(\bar{w})$ is mentioned in $\varphi^{+}_i(\bar{X},\bar{y})$}\}$.
One can easily check that for every assignment $(V, v)$ such that $(\A,V,v)\models \varphi_i'(\bar{X}, \bar{x}, \bar{y})$, it holds that $(\A,\bar{B}^v,v)\models \varphi_i'(\bar{X}, \bar{x}, \bar{y})$ and $\bar{B}^v \subseteq V(\bar{X})$.
Namely, $\bar{B}^v$ is a valid candidate for $\bar{X}$ and, furthermore, it is contained in all satisfying assignments of $\bar{X}$ when $v$ is fixed.
This offers an insight into the main idea of Part (1): by choosing one particular $v$ the plan is to remove $\bar{B}^v$ as an assignment over $\bar{X}$ in $\varphi_i$. 
For this, we choose the minimal $v$ that satisfies $\varphi^{\fo}_i(\bar{x},\bar{y})$ which can be defined with the following formula:
\[
\text{{\it min}-}\varphi^{\fo}_i(\bar{x}, \bar{y}) \; = \;  \,\varphi_i^{\fo}(\bar{x},\bar{y})\wedge \fa{\bar{x}'} \fa{\bar{y}'}\big(\varphi_i^{\fo}(\bar{x}',\bar{y}')\to ( \bar{x}\leq\bar{x}' \wedge  \bar{y}\leq\bar{y}')\big).
\] 
If $\varphi^{\fo}_i$ is satisfiable, let $v$ be the only assignment such that  $(\A,v)\models \text{{\it min}-}\varphi^{\fo}_i(\bar{x}, \bar{y})$. 
Furthermore, let $V^*$ be the second-order assignment and $v^*$ the first-order assignment that satisfy $V^*(\bar{X}) = \bar{B}^{v}$ and $v^*(\bar{x}) = v(\bar{x})$.
By the previous discussion, $(\A,V^*,v^*)\models\varphi_i(\bar{X},\bar{x})$.

Now, we have all the ingredients in order to define $\kappa(\alpha)$. 
Intuitively, we want to exclude the assignment $(V^*, v^*)$ from the satisfying assignments of $\varphi_i(\bar{X},\bar{x})$.
Towards this goal, we can define a formula $\psi_i(\bar{X},\bar{x})$ such that, if $\varphi_i(\bar{X}, \bar{x})$ is satisfiable, $(\A, V, v) \models \psi_i(\bar{X},\bar{x})$ if, and only if either $V \neq V^*$ or $v \neq v^*$. 
This property can be defined as follows:
\begin{align}
\psi_i(\bar{X},\bar{x}) \; := \;  &\big( \, \ex{\bar{x}'} \ex{\bar{y}} \varphi^{\fo}_i(\bar{x}',\bar{y}) \,\big) \rightarrow  \label{eq:part1} \\
&\Big( 
\ex{\bar{y}}\!\big[ \text{{\it min}-}\varphi^{\fo}_i(\bar{x}, \bar{y}) \wedge \big( \, \varphi_i'(\bar{X},\bar{x}, \bar{y})  \rightarrow \!\! \bigvee_{X \in \bar{X}}\! \ex{\bar{z}}\!( \, X(\bar{z}) \wedge   \!\!\!\!\!\!\!\!\!\!\!\! \bigwedge\limits_{X(\bar{w}) \, \in \, \varphi^{+}_i(\bar{X},\bar{y})} \!\!\!\!\!\!\!\!\!\!\!\! \bar{w}\neq\bar{z} \, ) \, \big)\big] \;  \vee \!  \label{eq:part2} \\ 
&\; \ex{\bar{x}'} \ex{\bar{y}}\big( \varphi_i'(\bar{X},\bar{x}',\bar{y})\wedge \bar{x}' < \bar{x}\big) \; \;\; \Big)  \label{eq:part3}
\end{align}

To understand the formula, first notice that the premise of the implication at (\ref{eq:part1}) is true if, and only if, $\varphi_i(\bar{X}, \bar{x})$ is satisfiable. 
Indeed, by Claim~\ref{claim:minusone} we know that if $\ex{\bar{x}} \ex{\bar{y}} \varphi^{\fo}_i(\bar{x},\bar{y})$ is true, then there exist assignments $V$ and $v$ such that $(\A,V,v) \models\varphi_i'(\bar{X},\bar{x},\bar{y})$.
The conclusion of the implication (divided into (\ref{eq:part2}) and (\ref{eq:part3})), takes care that $V(\bar{X}) \neq V^*(\bar{X})$ or $v(\bar{x}) \neq v^*(\bar{x})$.
Here, the first disjunct (\ref{eq:part2}) checks that $V(\bar{X}) \neq V^*(\bar{X})$ by defining that if $\varphi_i'(\bar{X},\bar{x}, \bar{y})$ is satisfied then $V^*(\bar{X}) \subsetneq V(\bar{X})$. 
The second disjunct (\ref{eq:part3}) is satisfied when $v(\bar{x})$ is not the lexicographically smallest tuple that satisfies $\varphi_i$ (i.e.\ $v(\bar{x}) \neq v^*(\bar{x})$).
Finally, from the previous discussion one can easily check that $\psi_i(\bar{X},\bar{x})$ satisfies the desired property.

We are ready to define the formula  $\kappa(\alpha)$ as $\sa{\bar{X}}\sa{\bar{x}} \bigvee_{i = 1}^m\varphi_i^*(\bar{X},\bar{x})$
where each modified disjunct $\varphi_i^*(\bar{X},\bar{x})$ is constructed as follows. 
For the sake of simplicity, define the auxiliary formula $\chi_i = \neg \ex{\bar{x}} \ex{\bar{y}}\varphi^{\fo}_i(\bar{x},\bar{y})$. 
This formula basically checks if $\varphi_i$ is not satisfiable (recall Claim~\ref{claim:minusone}).
Define the first formula $\varphi_1^*$ as:
\[
\varphi^*_1(\bar{X},\bar{x}) \; := \; \varphi_1(\bar{X},\bar{x})\wedge\psi_1(\bar{X},\bar{x}).
\]
This formula accepts all the assignments that satisfy $\varphi_1$, except for the assignment $(V^*,v^*)$ of $\varphi_1$. The second formula $\varphi_2^*$ is defined as:
\[
\varphi^*_2(\bar{X},\bar{x}) \; := \; \varphi_2(\bar{X},\bar{x})\wedge\psi_1(\bar{X},\bar{x})\wedge(\chi_1\to\psi_2(\bar{X},\bar{x})).
\]
This models all the assignments that satisfy $\varphi_2$, except for the assignment $(V^*,v^*)$ of $\varphi_1$. Moreover, if $\varphi_1$ is not satisfiable, then $\psi_1(\bar{X},\bar{x})$ and $\chi_1$ will hold, and the formula $\psi_2(\bar{X},\bar{x})$ will exclude the assignment $(V^*,v^*)$ of $\varphi_2$. 
This construction can be generalized for each $\varphi_i$ as follows:
\begin{multline*}
\varphi_i^*(\bar{X},\bar{x}) \; := \; \varphi_i(\bar{X},\bar{x})\wedge\psi_1(\bar{X},\bar{x})\,\wedge \\ (\chi_1\to\psi_2(\bar{X},\bar{x}))\wedge((\chi_1\wedge\chi_2)\to\psi_3(\bar{X},\bar{x}))\wedge\cdots\wedge(
\bigwedge_{j = 1}^{j = i-1}\chi_j\to\psi_i(\bar{X},\bar{x})).
\end{multline*}
From the construction of $\kappa(\alpha)$, one can easily check that  $\sem{\kappa(\alpha)}(\A) = \sem{\alpha}(\A)-1$ for each $\A$.

\medskip

\noindent {\em Part (2).} Let $\alpha = \beta + \sa{\bar{X}}\sa{\bar{x}} \varphi(\bar{X},\bar{x})$ for some $\eqso(\logex{1})$ formula $\beta$. We define $\kappa(\alpha)$ as follows.
First, rewrite $\varphi(\bar{X},\bar{x})$ as in Part (1). Let $\varphi = \bigvee_{i = 1}^m\varphi_i(\bar{X},\bar{x})$ where each $\varphi_i$ satisfies conditions (a), (b) and (c) defined above. Also, consider the previously defined formulae $\chi_i$ and $\psi_i$, for each $i \leq m$. 
We construct a function $\lambda$ that receives a formula $\beta \in \eqso(\logex{1})$ and produces a logic formula $\lambda(\beta)$ that satisfies $\A\models\lambda(\beta)$ if, and only if, $\sem{\beta}(\A) = 0$. If $\beta = \sa{\bar{x}} \varphi(\bar{x})$, then $\lambda(\beta) = \neg \ex{\bar{x}'} \varphi(\bar{x}')$. If $\beta = \sa{\bar{X}}\sa{\bar{x}}\varphi(\bar{X},\bar{x})$, then 
define $\lambda(\beta) = \chi_1\wedge \cdots\wedge\chi_m$. If $\beta = (\beta_1 + \beta_2)$, then $\lambda(\beta) = \lambda(\beta_1) \wedge \lambda(\beta_2)$.
Now, following the same ideas as in Part (1) we define a formula $\varphi_i^*(\bar{X},\bar{x})$ that removes the minimal $(V^*, v^*)$ of $\varphi_i$ whenever $\beta$ cannot be satisfied (i.e.\ $\lambda(\beta)$ is true). Formally, we define $\varphi_i^*$ as follows:
\begin{multline*}
\varphi_i^*(\bar{X},\bar{x}) \ := \ \varphi_i(\bar{X},\bar{x})\wedge\Big(\lambda(\beta)\to\Big(\psi_1(\bar{X},\bar{x})\,\wedge \\
(\chi_1\to\psi_2(\bar{X},\bar{x}))\wedge((\chi_1\wedge\chi_2)\to\psi_3(\bar{X},\bar{x}))\wedge\cdots\wedge(
\bigwedge_{j = 1}^{j = i-1}\chi_j\to\psi_i(\bar{X},\bar{x}))\Big)\Big).
\end{multline*}
Finally, $\kappa(\alpha)$ is defined as $\kappa(\alpha) = \kappa(\beta) + \sa{\bar{X}}\sa{\bar{x}} \bigvee_{i = 1}^m\varphi_i^*(\bar{X},\bar{x})$, which is in $\eqso(\logex{1})$ and satisfies the desired conditions. This concludes the proof.
%

\end{proof}
\end{leftbar}
%
%
The proof that $\eqso(\logex{1})$ is closed under subtraction by one is the most involved in the paper.
A key insight in this proof is the fact that if a formula in $\eqso(\logex{1})$ has a satisfying assignment over a structure, then it also has a satisfying assignment over this structure that is of logarithmic size, and which can be characterized and removed by using some fixed formulae. To give more intuition about this idea, let us consider the case of $\ctdnf$, and the way Saluja et\text{.} al \cite{SalujaST95}  propose to encode this problem in $\E{1}$ and, thus, also in $\eqso(\logex{1})$. Let $\R = \{P_0(\cdot,\cdot,\cdot), P_1(\cdot,\cdot,\cdot), P_2(\cdot,\cdot,\cdot), P_3(\cdot,\cdot,\cdot),<\}$ and $\theta$ be a propositional formula in 3-DNF. Then $\theta$ is encoded as an $\R$-structure $\A_\theta$ as follows. The domain of $\A_\theta$ is the set of propositional variables occurring in $\theta$, and for every tuple $(a,b,c)$ of propositional variables occurring in $\theta$, we have that: (i) $P_0(a,b,c)$ holds if $(a\wedge b \wedge c)$ is a disjunct in $\theta$;
(ii) $P_1(a,b,c)$ holds if $(a\wedge b \wedge \neg c)$ is a disjunct in $\theta$;
(iii) $P_2(a,b,c)$ holds if $(a\wedge \neg b \wedge \neg c)$ is a disjunct in $\theta$; and 
(iv) $P_3(a,b,c)$ holds if $(\neg a \wedge \neg b \wedge \neg c)$ is a disjunct in $\theta$.
Moreover, to define $\ctdnf$, we consider a fixed $\so$-formula $\varphi(T)$ over $\R$, where $T$ is a unary predicate, such that the number of satisfying assignments of $\theta$ is equal to $\sem{\sa{T}\varphi(T)}(\A_\theta)$. More specifically, $T(a)$ holds if and only if  $a$ is assigned value true, so that $\varphi(T)$ is defined as:
\begin{align*}
\varphi(T) \ := \ \ex{x}\ex{y}\ex{z}\big(\,
&(P_0(x,y,z)\wedge T(x)\wedge  T(y)\wedge T(z))\ \vee\\
&(P_1(x,y,z)\wedge T(x)\wedge  T(y)\wedge \neg T(z))\ \vee\\
&(P_2(x,y,z)\wedge T(x)\wedge \neg T(y)\wedge \neg T(z))\ \vee\\
&(P_3(x,y,z)\wedge \neg T(x)\wedge \neg T(y)\wedge \neg T(z))\big).
\end{align*}
Let us now focus on the first disjunct $(P_0(x,y,z)\wedge T(x)\wedge  T(y)\wedge T(z))$ of $\varphi(T)$. Assuming that a propositional formula $(a\wedge b \wedge c)$ is a disjunct in $\theta$, we have that $(\A_\theta,V) \models \varphi(T)$ for every assignment $V$ for $T$ such that $\{a,b,c\} \subseteq V(T)$. 
Notice that some of these assignments, such as the one that assigns to $T$ all the variables occurring in $\theta$, are of linear size in the size of $\A_\theta$. However, $\varphi(T)$ also admits satisfying assignments of logarithmic size in the size of $\A_\theta$, such as $V(T) = \{a,b,c\}$. A key idea in the proof is the fact that the minimum of such small witnesses (under the lexicographic order induced by~$<$) can be identified by using a fixed formula in $\logex{1}$:
\begin{multline*}
\alpha(x,y,z) \ := \ P_0(x,y,z) \wedge  \fa{x'}\fa{y'}\fa{z'} \big(P_0(x',y',z')\ \to\\
\big((x < x') \vee (x = x' \wedge y < y') \vee (x = x' \wedge y = y' \wedge z < z')\big)\big),
\end{multline*}
and then it can be removed also by using a fixed formula in $\logex{1}$:
\begin{align*}
\beta(T) \ := \ \, & \ex{x} \ex{y} \ex{z} \big(P_0(x,y,z)\wedge T(x)\wedge  T(y)\wedge T(z)\big)\ \wedge \\
& \ex{x'}\ex{y'}\ex{z'} \big(\alpha(x',y',z')\ \wedge \\
& \hspace{23mm} (T(x')\wedge T(y')\wedge T(z')\to \ex{w}(T(w) \wedge w \neq x' \wedge w\neq y' \wedge w'\neq z))\big).
\end{align*}
In particular, the combination of these two formulae forces $T$ to represent a satisfying assignment for $\theta$ that is different from the set $\{a_0,b_0,c_0\}$, where $(a_0,b_0,c_0)$ is the minimum tuple under the lexicographic order induced by~$<$ on the set of tuples $(a,b,c)$ such that $(a \wedge b \wedge c)$ is a disjunct of $\theta$. In the proof, we generalize and properly formalize this idea, thus using the existence of logarithmic size witnesses for the formulae in $\eqso(\logex{1})$ to prove that this class is closed under subtraction by one. We think that the existence of such witnesses is a fundamental property of this class that deserves to be further investigated. 



\subsection{Defining a class of functions with easy decision versions and natural complete problems}
\label{sec-horn}
\newcommand{\pP}{\textit{P}}
\newcommand{\pN}{\textit{N}}
\newcommand{\pV}{\textit{V}}
\newcommand{\pT}{\textit{T}}
\newcommand{\pA}{\textit{A}}
\newcommand{\pNC}{\textit{NC}}
\newcommand{\pD}{\textit{D}}

The goal of this section is to define a class of functions in $\shp$ with easy decision counterparts and natural complete problems. To this end, we consider the notion of parsimonious reduction. Formally, a function $f\colon\Sigma^* \to \N$ is parsimoniously reducible to a function $g\colon\Sigma^* \to \N$ if there exists a function $h\colon\Sigma^* \to \Sigma^*$ such that $h$ is computable in polynomial time and $f(x) = g(h(x))$ for every $x \in \Sigma^*$. As mentioned at the beginning of this section, if $f$ can be parsimoniously reduced to $g$, then $L_g \in \ptime$ implies that $L_f \in \ptime$ and the existence of an FPRAS for $g$ implies the existence of an FPRAS for $f$. 

In the previous section, we showed that the class $\eqso(\logex{1})$ has good closure and approximation properties. Unfortunately, it is not clear whether it admits a {\em natural} complete problem under parsimonious reductions, where {\em natural} means any of the counting problems defined in this section or any other well-known counting problem (not one specifically designed to be complete for the class). On the other hand, $\totp$ admits a natural complete problem under parsimonious reductions, which is the problem of counting the number of inputs accepted by a monotone circuit~\cite{BCPPZ17}. However, the notion of monotone circuit used in \cite{BCPPZ17} does not correspond with the usual notion of monotone circuit~\cite{GS90}, that is, circuits with AND and OR gates but without negation. In this sense, we still lack a class of functions in $\shp$ with easy decision counterparts and a complete problem that is well known and has been widely studied. In this section, we follow a different approach to find such a class,
which is inspired by the approach followed in \cite{G92} that uses a restriction of second-order logic to Horn clauses for capturing $\ptime$ (over ordered structures). The following example shows how our approach works.

\begin{exa} \label{ex-hornsat-esop1}
Let $\R = \{\pP(\cdot,\cdot), \pN(\cdot,\cdot), \pV(\cdot), \pNC(\cdot),<\}$. This vocabulary is used as follows to encode a Horn formula. A fact $\pP(c,x)$ indicates that propositional variable $x$ is a disjunct in a clause $c$, while $\pN(c,x)$ indicates that $\neg x$ is a disjunct in $c$. Furthermore, $\pV(x)$ holds if  $x$ is a propositional variable, and $\pNC(c)$ holds if $c$ is a clause containing only negative literals, that is, $c$ is of the form $(\neg x_1 \vee \cdots \vee \neg x_n)$.

To define $\chsat$, we consider an \so-formula $\varphi(\pT)$ over $\R$, where $\pT$ is a unary predicate, such that for every Horn formula $\theta$ encoded by an $\R$-structure $\A$, the number of satisfying assignments of $\theta$ is equal to $\sem{\sa{\pT} \varphi(\pT)}(\A)$. In particular, $\pT(x)$ holds if and only if $x$ is a propositional variable that is assigned value true.  More specifically, 
\begin{align*}
\varphi(\pT) \; :=\;\;  & \fa{x} (\pT(x) \to \pV(x)) \ \wedge\\
& \fa{c}  (\pNC(c) \to \ex{x} (\pN(c,x) \wedge \neg \pT(x))) \ \wedge\\
& \fa{c} \fa{x} ([\pP(c,x) \wedge \fa{y} (\pN(c,y) \to \pT(y))] \to \pT(x)).
\end{align*}
We can rewrite $\varphi(\pT)$ in the following way:
\begin{align*}
& \fa{x}  (\neg \pT(x) \vee \pV(x)) \ \wedge\\
& \fa{c}  (\neg \pNC(c) \vee \ex{x} (\pN(c,x) \wedge \neg \pT(x)))\ \wedge\\
& \fa{c} \fa{x}  (\neg \pP(c,x) \vee \ex{y} (\pN(c,y) \wedge \neg \pT(y)) \vee \pT(x)).
\end{align*}
Moreover, by introducing an auxiliary predicate $\pA$ defined as 
\begin{align*}
\fa{c} \fa{x}  (\neg \pA(c,x) \leftrightarrow [\pN(c,x) \wedge \neg \pT(x)]),
\end{align*}
we can translate $\varphi(\pT)$ into the following equivalent formula:
\begin{align*}
\psi(\pT,\pA) \; := \;\;  & \fa{x} (\neg \pT(x) \vee \pV(x)) \ \wedge\\
& \fa{c} (\neg \textit{NC}(c) \vee \ex{x} \neg \textit{A}(c,x)) \ \wedge\\
& \fa{c} \fa{x}  (\neg \textit{P}(c,x) \vee \ex{y} \neg \textit{A}(c,y) \vee \textit{T}(x)) \ \wedge\\
& \fa{c} \fa{x} (\neg \textit{N}(c,x) \vee \textit{T}(x) \vee \neg \textit{A}(c,x)) \ \wedge \\
& \fa{c} \fa{x} (\textit{A}(c,x) \vee \textit{N}(c,x)) \ \wedge\\
& \fa{c} \fa{x} (\textit{A}(c,x) \vee \neg\textit{T}(x)).
\end{align*}
More precisely, we have that:
\begin{align*}
\sem{\sa{\pT} \varphi(\pT)}(\A) &= \sem{\sa{\pT} \sa{\pA} \psi(\pT,\pA)}(\A),
\end{align*}
 for every $\R$-structure $\A$ encoding a Horn formula. Therefore, the formula $\psi(\pT,\pA)$ also defines $\chsat$. More importantly, $\psi(\pT,\pA)$ resembles a conjunction of Horn clauses except for the use of negative literals of the form $\ex{v} \neg \textit{A}(u,v)$. \qed
\end{exa}
The previous example suggests that to define $\chsat$, we can use Horn formulae defined as follows. 
A positive literal is a formula of the form $X(\x)$, where $X$ is a second-order variable and $\x$ is a tuple of first-order variables, and a negative literal is a formula of the form $\ex{\v} \neg X(\u,\v)$, where $\u$ and $\v$ are tuples of first-order variables. Given a signature $\R$, a clause over $\R$ is a formula of the form $\fa{\x} (\varphi_1 \vee \cdots \vee \varphi_n)$, 
where each $\varphi_i$ ($1 \leq i \leq n$) is either a positive literal, a negative literal or an \fo-formula over $\R$.  A clause is said to be Horn if it contains at most one positive literal, and a formula is said to be Horn if it is a conjunction of Horn clauses. With this terminology, we define $\uhorn$ as the set of formulae $\psi$ such that $\psi$ is a Horn formula over a signature $\R$. 

As we have seen, we have that $\chsat \in \eqso(\uhorn)$. Moreover, one can show that $\eqso(\uhorn)$ forms a class of functions with easy decision counterparts, namely, $\eqso(\uhorn) \subseteq \totp$.
Thus, $\eqso(\uhorn)$ is a new alternative in our search for a class of functions in $\shp$ with easy decision counterparts and natural complete problems. Moreover, an even larger class for our search can be generated by extending the definition of $\uhorn$ with outermost existential quantification. 
Formally, a formula $\varphi$ is in $\ehorn$ if $\varphi$ is of the form $\ex{\bar x} \psi$ with $\psi$ a Horn formula. 

\begin{prop}\label{prop:ehorn-pe}
$\eqso(\ehorn) \subseteq \totp$.
\end{prop}
In this section, we identify a complete problem for $\eqso(\ehorn)$ under parsimonious reductions. Hence, to prove that $\eqso(\ehorn) \subseteq \totp$, it is enough to prove that such a problem is in $\totp$, as $\totp$ is closed under parsimonious reductions. We give this proof at the end of this section, after the complete problem has been identified.

Interestingly, we have that both $\chsat$ and $\cdnf$ belong to $\eqso(\ehorn)$. 
An imperative question at this point is whether in the definitions of $\uhorn$ and $\ehorn$, it is necessary to allow negative literals of the form $\ex{\v} \neg X(\u,\v)$. Actually, this forces our Horn classes to be included in $\eqso(\logu{2})$ and not necessarily in $\eqso(\loge{2})$. The following result shows that this is indeed the case.

\begin{prop}\label{prop:hsat-not-sigma2}	
$\chsat \not\in \eqso(\loge{2})$.
\end{prop}
\begin{proof}
Suppose that the statement is false, that is, $\chsat \in \eqso(\loge{2})$. Consider the signature $\R$ from Example~\ref{ex-hornsat-esop1} and let $\alpha \in \eqso(\loge{2})$ be a formula over $\R$ that defines $\chsat$. By Proposition~\ref{theo-pnf-snf} we know that every formula in $\eqso(\loge{2})$ can be rewritten in $\loge{2}$-PNF, so we can assume that $\alpha$ is of the form $\sa{\bar{X}}\sa{\bar{x}} \ex{\bar{y}} \fa{\bar{z}}\varphi(\bar{X},\bar{x},\bar{y},\bar{z})$. Now, consider the following Horn formula:
$$
\Phi \ = \ p \wedge \bigwedge_{i = 1}^n (t_i \wedge p \to q) \wedge \neg q,
$$
such that $n = \length{\bar{x}} + \length{\bar{y}} + 1$ and let $\A_{\Phi}$ be the encoding of this formula over $\R$. 
One can easily check that $\Phi$ is satisfiable, so $\sem{\alpha}(\A_{\Phi}) \geq 1$. Let $(\bar{B},\bar{b},\bar{a})$ be an assignment to $(\bar{X},\bar{x},\bar{y})$ such that $\A_{\Phi} \models \fa{\bar{z}} \varphi(\bar{B},\bar{b},\bar{a},\bar{z})$ and let $t_{\ell}$ be such that it does not appear in $\bar{b}$ or $\bar{a}$ (recall that $n > \length{\bar{x}} + \length{\bar{y}}$). Consider the induced substructure $\A_{\Phi}'$ that is obtained by removing $t_{\ell}$ from $\A_{\Phi}$ and $\bar{B}'$ as the subset of $\bar{B}$ obtained by deleting each appearance of $t_{\ell}$ in $\bar{B}$. Given that if a universal formula holds in a structure $\A$, then it holds in every induced substructure of $\A$, we have that $\A_{\Phi}'\models \fa{\bar{z}}\varphi(\bar{B}',\bar{b},\bar{a},\bar{z})$. And so, it follows that $\sem{\alpha}(\A_{\Phi}') \geq 1$ which is not possible since $\A_{\Phi}'$ encodes the formula
$$
\Phi' = p \wedge \bigwedge_{i = 1}^{\ell-1} (t_i \wedge p \to q) \wedge (p\to q) \wedge \bigwedge_{i = \ell+1}^{n} (t_i \wedge p \to q) \wedge \neg q,
$$
which is unsatisfiable. This leads to a contradiction and we conclude that $\chsat$ is not in $\eqso(\ehorn)$.

\end{proof}
Next we show that $\eqso(\ehorn)$ is the class we were looking for, as not only every function in $\eqso(\ehorn)$ has an easy decision counterpart, but also $\eqso(\ehorn)$ admits a natural complete problem under parsimonious reductions. More precisely, define 
$\shdhsat$ as the problem of counting the satisfying assignments of a formula $\Phi$ that is a disjunction of Horn formulae. Then we have that:

\begin{thm} \label{sigma2hard}
	$\shdhsat$ is $\eqso(\ehorn)$-complete under parsimonious reductions. 
\end{thm}
\begin{proof}

First we prove that $\shdhsat$ is in $\eqso(\ehorn)$. Recall that each instance of $\shdhsat$ is a disjunction of Horn formulae. Let $\R$ be a relational signature such that $\R = \{\pP(\cdot,\cdot), \pN(\cdot,\cdot), \pV(\cdot), \pNC(\cdot), \pD(\cdot,\cdot)\}$. Each symbol in this vocabulary is used to indicate the same as in Example \ref{ex-hornsat-esop1}, with the addition of $\pD(d,c)$ which indicates that $c$ is a clause in the formula $d$. Define $\psi$ as in Example \ref{ex-hornsat-esop1} such that $\sa{\pT} \sa{\pA} \psi(\pT,\pA)$
defines $\chsat$. In order to encode $\shdhsat$, we extend $\psi(\pT,\pA)$ by adding the information of $\pD(d,c)$ as follows:
\begin{align*}
\psi'(T,A) \ := \ \ex{d} \big[ \ & \fa{x} (\neg \pT(x) \vee \pV(x)) \ \wedge\\
& \fa{c} (\neg \pD(c,d)\vee \neg \textit{NC}(c) \vee  \ex{x} \neg \textit{A}(c,x)) \ \wedge\\
& \fa{c} \fa{x} (\neg \pD(c,d)\vee\neg \textit{P}(c,x) \vee \ex{y} \neg \textit{A}(c,y) \vee \textit{T}(x)) \ \wedge\\
& \fa{c} \fa{x} (\neg \pD(c,d)\vee\neg \textit{N}(c,x) \vee \textit{T}(x) \vee \neg \textit{A}(c,x)) \ \wedge\\
& \fa{c} \fa{x}  (\neg \pD(c,d)\vee\textit{A}(c,x) \vee \textit{N}(c,x)) \ \wedge\\
& \fa{c} \fa{x} (\neg \pD(c,d)\vee\textit{A}(c,x) \vee \neg\textit{T}(x)) \ \big].
\end{align*}
One can check that $\psi'(T,A)$ effectively defines $\shdhsat$ as for every disjunction of Horn formulae $\theta = \theta_1\vee\cdots\vee\theta_m$ encoded by an $\R$-structure $\A$, the number of satisfying assignments of $\theta$ is equal to $\sem{\sa{\pT} \sa{\pA} \psi'(\pT,\pA)}(\A)$.  Therefore, we conclude that $\shdhsat \in \eqso(\ehorn)$.

Next, we prove that $\shdhsat$ is hard for $\eqso(\ehorn)$ over each signature~$\R$ under parsimonious reductions. For each $\eqso(\ehorn)$ formula $\alpha$ over $\R$, we define a polynomial-time function $g_{\alpha}$ that receives an $\R$-structure $\A$ and outputs an instance of $\shdhsat$ such that $\sem{\alpha}(\A) = \shdhsat(g_{\alpha}(\A))$. By Proposition~\ref{theo-pnf-snf}, we can assume that $\alpha$ is of the form:
$$
\alpha \ = \ \sum_{i = 1}^{m} \sa{\bar{X}_{i}}\sa{\bar{x}} \ex{\bar{y}} \bigwedge_{j = 1}^{n} \fa{\bar{z}} \varphi^i_j(\bar{X}_{i},\bar{x},\bar{y},\bar{z}),
$$
where each $\varphi^i_j$ is a Horn clause, and each $\bar{X}_{i}$ is a sequence of second-order variables.
Consider $\bar{X}$ as the union of all $\bar{X}_{i}$. We replace each of the $m$ summands in $\alpha$ with
$$
\sa{\bar{X}}\sa{\bar{x}} \ex{\bar{y}} \bigg( \bigwedge_{j = 1}^{n} \fa{\bar{z}} \varphi^i_j(\bar{X}_{i},\bar{x},\bar{y},\bar{z})\wedge\bigwedge_{X\not\in\bar{X}_{i}}\!\!\fa{\bar{u}}X(\bar{u})\bigg),
$$
whose sum is equivalent to $\alpha$.
Now, consider a finite $\R$-structure $\A$ with domain $A$. 
The next transformation of $\alpha$ and $\A$ towards a disjunction of Horn-formulae is to expand each first-order quantifier (i.e.\ $\Sigma{\bar{x}}$,  $\exists\bar{y}$, and $\forall\bar{z}$) by replacing variables with constants.
More specifically, we obtain the following formula that defines the same function as $\alpha$, and it is of polynomial size in the size of $\A$ (recall that $\alpha$ is fixed):
$$
\alpha_{\A} \ = \ \sum_{i = 1}^{m} \sum_{\bar{a}\,\in A^{\length{\bar{x}}}} \sa{\bar{X}}\bigvee_{\bar{b}\,\in A^{\length{\bar{y}}}}\bigg(\bigwedge_{j = 1}^n\bigwedge_{\bar{c}\,\in A^{\length{\bar{z}}}}\varphi^i_j(\bar{X}_i,\bar{a},\bar{b},\bar{c})\wedge \bigwedge_{X \not\in \bar{X}_i}\bigwedge_{\bar{e}\,\in A^{\arity(X)}}\!\!\!\!\!\!\! X(\bar{e})\bigg).
$$

Notice that each first-order subformula in $\varphi^i_j(\bar{X}_i,\bar{a},\bar{b},\bar{c})$ has no free variables and, therefore, we can evaluate each of them in polynomial time and easily rewrite $\alpha_{\A}$ to an equivalent formula that does not have any first-order subformula. In other words, in polynomial time we can replace $\varphi^i_j$ with a disjunction of negative literals of the form $\neg X_{\ell}(\bar d)$ and at most one positive literal of the form $X_{\ell}(\bar d)$, where $\bar d$ is a tuple of constants. After simplifying, grouping and reordering the previous formula, we can obtain an equivalent formula $\alpha_{\A}'$ of the form:
$$
\alpha_{\A}' \ = \ \sum_{i = 1}^{m'}\sa{\bar{X}}\bigvee_{j = 1}^{n_1'}\bigwedge_{k = 1}^{n_2'}\psi^{i}_{j,k}(\bar{X}),
$$
where every $\psi^{i}_{j,k}(\bar{X})$ is a disjunction of 
negative literals of the form $\neg X_{\ell}(\bar d)$ and at most one positive literal of the form $X_{\ell}(\bar d)$, where $\bar d$ is a tuple of constants. 

The idea for the rest of the proof is to show how to obtain $g_{\alpha}(\A)$, i.e.\ an instance of $\shdhsat$, from $\alpha_{\A}'$.
First, if $m' = 1$ and $\alpha_{\A}' = \sa{\bar{X}}\bigvee_{j = 1}^{n_1'}\bigwedge_{k = 1}^{n_2'}\psi_{j,k}(\bar{X})$, then we can define $g_{\alpha}(\A)$ as the propositional formula $\bigvee_{j = 1}^{n_1'}\bigwedge_{k = 1}^{n_2'}\psi_{j,k}(\bar{X})$ over the propositional alphabet $\{X(\bar{e}) \mid X \in \bar{X} \text{ and } \bar{e}\in A^{\arity(X)} \}$. It is straightforward to see that $\bigvee_{j = 1}^{n_1'}\bigwedge_{k = 1}^{n_2'}\psi_{j,k}(\bar{X})$ is a disjunction of Horn formulae, and its number of satisfying assignments is exactly $\sem{\alpha}(\A)$. 
Otherwise, if $m' > 1$, then we can use $m'$ fresh new variables $t_1,\ldots,t_{m'}$ and define:
$$
g_{\alpha}(\A) \ := \ \bigvee_{i=1}^{m'} \, \bigvee_{j = 1}^{n_1'} \, \bigwedge_{k = 1}^{n_2'} \, \psi^i_{j,k}(\bar{X}) \wedge t_i \wedge \bigwedge_{\ell \neq i} \neg t_{\ell}
$$ 
over the propositional alphabet $\{X(\bar{e}) \mid X \in \bar{X} \text{ and } \bar{e}\in A^{\arity(X)} \} \cup \{t_1,\ldots,t_{m'}\}$.
Variables $t_1,\ldots,t_{m'}$ are used to have disjoint sets of propositional assignments for the different disjuncts of the outermost disjunction, which correspond to the summands in the original formula.
One can easily check that $g_{\alpha}(\A)$ is a disjunction of Horn formulae, and that the number of satisfying assignments of $g_{\alpha}(\A)$ is exactly $\sem{\alpha}(\A)$. This covers all possible cases for $\alpha$, and the entire procedure takes polynomial time.
%
%

\end{proof}
Now that we have a complete problem for $\eqso(\ehorn)$, we can provide a simple proof of Proposition~\ref{prop:ehorn-pe}.

\medskip

\noindent{\emph{Proof of Proposition~\ref{prop:ehorn-pe}.}}
%
%
As we mentioned before, $\totp$ is closed under parsimonious reductions, so we only need to show that 
$\shdhsat$ is in $\totp$. For this consider a non-deterministic procedure that receives a $\dhsat$ formula $\Phi$ as input and does the following. First it checks whether $\Phi$ is satisfiable. If it is not, then it stops; otherwise, it creates a dummy branch that simply stops, and continues in the main branch. More precisely, it picks in the main branch a propositional variable $x$ in $\Phi$ and creates two formulae $\Phi_0$ and $\Phi_1$, where $x$ has been replaced by $\perp$ and $\top$, respectively. If only one of these is satisfiable, it continues on this branch with the respective $\Phi_i$, and if both are satisfiable, then it creates a new branch for $\Phi_1$ and continues on this branch with $\Phi_0$. On each branch, it repeats the same instructions until no variables are left to replace. Since $\dhsat$ is in $\ptime$, all the aforementioned checks can be done in polynomial time, so that the procedure takes at most $h(n)$ steps in each branch, where $n$ is the size of $\Phi$ and $h$ is some fixed polynomial. Moreover, the algorithm produces exactly $\shdhsat(\Phi)+1$ branches, from which we conclude that $\shdhsat$ is in $\totp$. 
%

\qed

Finally, it is important to mention that from the previous proof one can easily derive that $\eqso(\ehorn) \equiv \#(\ehorn)$. Therefore, the framework in~\cite{SalujaST95} is enough for defining the class of problems in $\eqso(\ehorn)$.


\section{Adding recursion to QSO}\label{sec:beyond}

We have used weighted logics to give a framework for descriptive complexity of counting complexity~classes. Here, we go beyond weighted logics and give the first steps on defining recursion at the quantitative level.
This goal is not trivial not only because we want to add recursion over functions, but also because 
it is not clear what could be the right notion of ``fixed point''. 
To this end, we show first how to extend $\qso$ with function symbols that are later
used 
to define a natural generalization of LFP to functions.
 As a proof of concept, we show that this notion can be used to capture $\fp$.
Moreover, we use this concept
to define an operator for counting paths in a graph, a natural generalization of the transitive closure operator~\cite{immerman1999descriptive}, and show that this gives rise to a logic that captures~$\shl$. 

We start by defining an extension of $\qso$ with function symbols. Assume that $\fs$ is an infinite set of function symbols, where each $h \in \fs$ has an associated arity denoted by $\arity(h)$. Then the set of $\fqso$ formulae over a signature $\R$ is defined by the following grammar:
\begin{multline}
\label{eq-fqso}
	\alpha \ :=  \ \varphi \ \mid \  s \  \mid \  h(x_1, \ldots, x_\ell) \  \mid \
	(\alpha \add \alpha) \  \mid\  (\alpha \mult \alpha) \  \mid \\
	\sa{x} \alpha \  \mid \
	\pa{x} \alpha \  \mid \
	\sa{X} \alpha \  \mid \
	\pa{X} \alpha,
\end{multline}
where $h \in \fs$, $\arity(h) = \ell$ and $x_1, \ldots, x_\ell$ is a sequence of (not necessarily distinct) first-order variables. Given an $\R$-structure $\A$ with domain $A$, we say that $F$ is a \emph{function assignment} for $\A$ if for every $h \in \fs$ with $\arity(h) = \ell$, we have that $F(h)\colon A^\ell \to \N$. The notion of function assignment is used to extend the semantics of $\qso$ to the case of a quantitative formula of the form $h(x_1, \ldots, x_\ell)$. More precisely, given first-order and second-order assignments $v$ and $V$ for $\A$, respectively, 
we have that:
\begin{align*}
\sem{h(x_1, \ldots, x_\ell)}(\A,v,V,F) &= F(h)(v(x_1),\ldots, v(x_\ell)).
\end{align*}
As for the case of $\qfo$, we define $\fqfo$ disallowing quantifiers $\Sigma X$ and $\Pi X$ in \eqref{eq-fqso}.

It is worth noting that function symbols in $\fqso$ represent functions from tuples to natural numbers, so they are different from the classical notion of function symbol in $\fo$~\cite{L04}. 
Furthermore, a function symbol can be seen as an ``oracle'' that is instantiated by the function assignment. 
To the best of our knowledge, this is the first article to propose this extension on weighted logics, which we think should be further investigated. 

We define an extension of LFP \cite{I86,vardi1982complexity} to allow counting. 
More precisely, the set of $\rqfo(\fo)$ formulae over a signature $\R$, where $\rqfo$ stands for recursive $\qfo$, is defined as an extension of $\qfo(\fo)$ that includes the formula $\clfp{\beta(\x, h)}$, where (1) $\x = (x_1, \ldots, x_\ell)$ is a sequence of $\ell$ distinct first-order variables, (2) $\beta(\x, h)$ is an $\fqfo(\fo)$-formula over $\R$ whose only function symbol is $h$, and (3) $\arity(h) = \ell$. The free variables of the formula $\clfp{\beta(\x,h)}$ are $x_1, \ldots, x_\ell$; in particular, $h$ is not considered to be free.

Fix an $\R$-structure with domain $A$ and a quantitative formula $\clfp{\beta(\x,h)}$ with $\arity(h) = \ell$, and assume that $\F$ is the set of functions $f\colon A^\ell \to \N$. To define the semantics of $\clfp{\beta(\x,h)}$, we first show how $\beta(\x,h)$ can be interpreted as an operator $T_{\beta}$ on $\F$. More precisely, for every $f \in \F$ and tuple $\a = (a_1, \ldots, a_{\ell}) \in A^\ell$, the function $T_{\beta}(f)$ satisfies that:
\begin{align*}
T_{\beta}(f)(\a) &= \sem{\beta(\x, h)}(\A,v,F),
\end{align*}
where $v$ is a first-order assignment  for $\A$ such that $v(x_i) = a_i$ for every $i \in \{1, \ldots, \ell\}$, and $F$ is a function assignment for $\A$ such that $F(h) = f$. 

As for the case of LFP, it would be natural to consider the point-wise partial order $\leq$ on $\F$ defined as $f \leq g$ if, and only if, $f(\bar{a}) \leq g(\bar{a})$ for every $\bar{a} \in A^{\ell}$, and let the semantics of $\clfp{\beta(\x,h)}$ be the least fixed point of the operator $T_\beta$. However, $(\F, \leq)$ is not a complete lattice, so we do not have a Knaster-Tarski Theorem ensuring that such a fixed point exists. Instead, we generalize the semantics of LFP as follows. In the definition of the semantics of LFP, an operator $T$ on relations is considered, and the semantics is defined in terms of the least fixed point of $T$, that is, a relation $R$ such that~\cite{I86,vardi1982complexity}: 
(a) $T(R) = R$, and (b) $R \subseteq S$ for every $S$ such that $T(S) = S$. 
We can view $T$ as an operator on functions if we consider the characteristic function of a relation. Given a relation $R \subseteq A^\ell$, let $\chi_R$ be its characteristic function, that is $\chi_R(\bar a) = 1$ if $\bar a \in R$, and $\chi_R(\bar a) = 0$ otherwise. Then define an operator $T^\star$ on characteristic functions as $T^\star(\chi_R) = \chi_{T(R)}$. Moreover, we can rewrite the conditions defining a least fixed point of $T$ in terms of the operator $T^\star$ if we consider the notion of support of a function. Given a function $f \in \F$, define the support of $f$, denoted by $\support(f)$, as $\{ \bar a \in A^\ell \mid f(\bar a) > 0 \}$. Then given that $\support(\chi_R) = R$, we have that the conditions (a) and (b) are equivalent to the following conditions on $T^\star$:
(a) $\support(T^\star(\chi_R)) = \support(\chi_R)$, and  (b) $\support(\chi_R) \subseteq \support(\chi_S)$ for every $S$ such that  $\support(T^\star(\chi_{S})) = \support(\chi_S)$.
To define a notion of fixed point for $T_\beta$ we simply generalize these conditions. More precisely, a function $f \in \F$ is a {\em s-fixed point} of $T_{\beta}$ if $\support(T_\beta(f)) = \support(f)$, and $f$ is a {\em least s-fixed point} of $T_{\beta}$ if $f$ is a s-fixed point of $T_\beta$ and for every s-fixed point $g$ of $T_\beta$ it holds that $\support(f) \subseteq \support(g)$. The existence of such fixed point is ensured by the following lemma:
\begin{lem}\label{lem-support}
Let $h \in \fs$ such that $\arity(h) = \ell$, and $\beta$ be an $\fqfo(\fo)$-formula over a signature $\R$ such that if a function symbol occurs in $\beta$, then this function symbol is $h$. Moreover, let $\A$ be an $\R$-structure with domain $A$, $f,g\colon A^\ell \to \mathbb{N}$ and $F,G$ be function assignments such that $F(h) = f$ and $F(h) = g$. If $\support(f) \subseteq \support(g)$, then for every first-order and second-order assignments $v$ and $V$, respectively, it holds that:
\begin{center}
if $\sem{\beta}(\A,v,V,F) > 0$, then $\sem{\beta}(\A,v,V,G) > 0$.
\end{center}
\end{lem}
\begin{proof}

We prove the lemma by induction on the structure of $\beta$. First we need to consider the base cases.
\begin{enumerate}
\item Assume that $\beta$ is either a constant $s \in \mathbb{N}$ or an $\fo$-formula $\varphi$. In both cases, function symbol $h$ is not mentioned, so $\sem{\beta}(\A,v,V,F) = \sem{\beta}(\A,v,V,G)$ and it trivially holds that if $\sem{\beta}(\A,v,V,F) > 0$, then $\sem{\beta}(\A,v,V,G) > 0$.

\item Assume that $\beta$ is equal to $h(y_1, \ldots, y_\ell)$, where $y_1$, $\ldots$, $y_\ell$ is a sequence of (non-necessarily pairwise distinct) variables. Let $\bar a = (v(y_1), \ldots, v(y_\ell))$. Then we have that $\sem{\beta}(\A,v,V,F) = F(h)(v(y_1), \ldots, v(y_\ell)) = f(\bar a)$ and $\sem{\beta}(\A,v,V,G) = g(\bar a)$. Given that $\supp(f) \subseteq \supp(g)$, if $f(\bar a) > 0$, then $g(\bar a) > 0$. Hence, we conclude that if $\sem{\beta}(\A,v,V,F) > 0$, then $\sem{\beta}(\A,v,V,G) > 0$.
\end{enumerate}
We now consider the inductive steps. Assume that the property holds for $\fqfo(\fo)$-formulae $\beta_1$, $\beta_2$ and $\delta$
\begin{enumerate}
\setcounter{enumi}{2}
\item Assume that $\beta = (\beta_1 + \beta_2)$. If $\sem{\beta}(\A,v,V,F) > 0$, then 
$\sem{\beta_1}(\A,v,V,F) > 0$ or $\sem{\beta_2}(\A,v,V,F) > 0$. Thus, by induction hypothesis we conclude that $\sem{\beta_1}(\A,v,V,G) > 0$ or $\sem{\beta_2}(\A,v,V,G) > 0$. Hence, we have that $\sem{\beta}(\A,v,V,G) > 0$.

\item Assume that $\beta = (\beta_1 \mult \beta_2)$. If $\sem{\beta}(\A,v,V,F) > 0$, then 
$\sem{\beta_1}(\A,v,V,F) > 0$ and $\sem{\beta_2}(\A,v,V,F) > 0$. Thus, by induction hypothesis we conclude that $\sem{\beta_1}(\A,v,V,G) > 0$ and $\sem{\beta_2}(\A,v,V,G) > 0$. Hence, we have that $\sem{\beta}(\A,v,V,G) > 0$.


\item Suppose that $\beta = \sa{x} \delta$. Then we have that $\sem{\beta}(\A,v,V,F) = \sum_{a \in A} \sem{\delta}(\A,v[a/x],V,F)$ and $\sem{\beta}(\A,v,V,G) = \sum_{a \in A} \sem{\delta}(\A,v[a/x],V,G)$. Thus, if we assume that $\sem{\beta}(\A,v,V,F) > 0$, then there exists $a \in A$ such that $\sem{\delta}(\A,v[a/x],V,F) > 0$. Hence, by induction hypothesis we have that $\sem{\delta}(\A,v[a/x],V,G) > 0$ and, therefore, we conclude that $\sem{\beta}(\A,v,V,G) >0$.


\item Suppose that $\beta = \pa{x} \delta$. Then we have that $\sem{\beta}(\A,v,V,F) = \prod_{a \in A} \sem{\delta}(\A,v[a/x],V,F)$ and $\sem{\beta}(\A,v,V,G) = \prod_{a \in A} \sem{\delta}(\A,v[a/x],V,G)$. Thus, if we assume that $\sem{\beta}(\A,v,V,F) > 0$, then $\sem{\delta}(\A,v[a/x],V,F) > 0$ for every $a \in A$. Hence, by induction hypothesis we have that $\sem{\delta}(\A,v[a/x],V,G) > 0$ for every $a \in A$, and, therefore, we conclude that $\sem{\beta}(\A,v,V,G) >0$.
%
  \qedhere
\end{enumerate}
%

\end{proof}

In the particular case of an $\rqfo(\fo)$-formula $\clfp{\beta(\x, h)}$, Lemma  \ref{lem-support} tell us that if $f,g \in \F$ and $\support(f) \subseteq \support(g)$, then $\support(T_\beta(f)) \subseteq \support(T_\beta(g))$. Hence, as for the case of LFP, this lemma gives us a simple way to compute a least s-fixed point of $T_\beta$. Let $f_0 \in \F$ be a function such that $f_0(\bar a) = 0$ for every $\bar a \in A^\ell$ (i.e.\ $f_0$ is the only function with empty support), and let function $f_{i+1}$ be defined as $T_\beta(f_i)$ for every $i \in \N$. Then there exists $j \geq 0$ such that $\support(f_j) = \support(T_\beta(f_j))$. Let $k$ be the smallest natural number such that $\support(f_{k}) = \support(T_\beta(f_k))$. We have that $f_k$ is a least s-fixed point of $T_\beta$, which is used to define the semantics of $\clfp{\beta(\x, h)}$. More specifically, for an arbitrary first-order assignment $v$ for $\A$:
\begin{align*}
\sem{\clfp{\beta(\x, h)}}(\A,v) &= f_{k}(v(\x)).
\end{align*}

\begin{exa} \label{ex:count-path}
We would like to define an $\rqfo(\fo)$-formula that, given a directed acyclic graph $G$ with $n$ nodes and a pair of nodes $b$, $c$ in $G$, counts the number of paths of length less than $n$ from $b$ to $c$ in $G$. To this end, assume that graphs are encoded using the signature $\R = \{ E(\cdot,\cdot),<\}$, and then define formula $\alpha(x, y, f)$ as follows:
\begin{eqnarray}\label{eq-count-paths-acyclic}
E(x,y) + \sa{z} f(x,z)\cdot E(z,y).
\end{eqnarray}
We have that $\clfp{\alpha(x,y,f)}$ defines our counting function. In fact, assume that $\A$ is an $\R$-structure with $n$ elements in its domain encoding an acyclic directed graph. Moreover, assume that $b,c$ are elements of $\A$ and $v$ is a first-order assignment over $\A$ such that $v(x) = b$ and $v(y) = c$. Then we have that $\sem{\clfp{\alpha(x,y,f)}}(\A,v)$ is equal to the number of paths in $\A$ from $b$ to $c$ of length at most $n-1$.

Assume now that we need to count the number of paths of length less than $n$ from a node $b$ to a node $c$ in a directed graph that is not necessarily acyclic. 
Then we cannot use formula \eqref{eq-count-paths-acyclic}, as the fixed point could be reached too soon without counting all paths of length at most $n-1$. The fixed point of \eqref{eq-count-paths-acyclic} is always reached after $k$ steps, where $k$ is the size of the maximum shortest path between two vertices in the graph. If the graph has a cycle in it, there might be paths of size $k' \in \{k+1,\ldots,n-1\}$ that would not be counted because the fixed point was reached earlier. Thus, we need a more involved formula in the general case, which is given below. 

Suppose that $\varphi_{\text{\rm first}}(x)$ and $\varphi_{\text{succ}}(x,y)$ are $\fo$-formulae defining the first element of $<$ and the successor relation associated to $<$, respectively.
Moreover, define formula $\beta(x, y, t, g)$ as follows:
\begin{align*}
(E(x,y) + \sa{z} g(x,z,t)\cdot E(z,y)) \cdot \varphi_{\text{\rm first}}(t) \ +
\sa{t'} \varphi_{\text{succ}}(t',t) \cdot \left(\sa{x'} \sa{y'} g(x',y',t') \right)
\end{align*}
Then our extended counting function is defined by:
$$
\sa{t} (\varphi_{\text{\rm first}}(t) \cdot \clfp{\beta(x,y,t,g)}).
$$ 
In fact, the number of paths of length at most $n$ from a node $x$ to a node $y$ is recursively computed by using the formula $(E(x,y) + \sa{z} g(x,z,t)\cdot E(z,y)) \cdot \varphi_{\text{\rm first}}(t)$, which stores this value in $g(x,y,t)$ with $t$ the first element in the domain.  The other formula $\sa{t'} \varphi_{\text{succ}}(t',t) \cdot \left(\sa{x'} \sa{y'} g(x',y',t') \right)$ is just an auxiliary artifact that is used as a counter to allow reaching a fixed point in the support of $g$ in $n$ steps. Notice that the use of the filter $\varphi_{\text{succ}}(t',t)$ prevents this formula for incrementing the value of $g(x,y,t)$ when $t$ is the first element in the domain.
\qed
\end{exa}

In contrast to $\lfp$, to reach a fixed point we do not need to impose any positive restriction on the formula $\beta(\x,h)$.
In fact, since $\beta$ is constructed from monotone operations (sum and product) over the natural numbers, the resulting operator $T_{\beta}$ is monotone as well.

Now that a least fixed point operator over functions is defined, the next step is to understand its expressive power.
In the following theorem, we show that this operator can be used to capture $\fp$.
\begin{thm} \label{rqfo-fo-cap}
	$\rqfo(\fo)$ captures $\fp$ over ordered structures.
\end{thm}
\begin{proof}

\newcommand{\ttB}{\mathtt{B}}
\newcommand{\successor}{\text{succ}}
\newcommand{\out}{\text{\it out}}

Given the definition of the semantics of $\rqfo(\fo)$, it is clear that a fixed formula $\clfp{\beta(\x, h)}$ can be evaluated in polynomial time, from which it is possible to conclude that each fixed formula in $\rqfo(\fo)$ can be evaluated in polynomial time. Thus, to prove that $\rqfo(\fo)$ captures $\fp$, we only need to prove the second condition in Definition \ref{def:cap}.

Let $f$ be a function in $\fp$. We address the case when $f$ is defined for the encodings of the structures of a relational signature $\R = \{ E(\cdot, \cdot) \}$, as the proof for an arbitrary signature is analogous.
 Let $M$ be a deterministic polynomial-time TM with a working tape and an output tape, such that the output of $M$ on input $\enc(\A)$ is $f(\enc(\A))$ for each $\R$-structure $\A$. We assume that $M = (Q,\{0,1\},q_0,\delta)$, 
 where $Q = \{q_0,\ldots,q_{\ell}\}$, and $\delta\colon Q\times\{0,1,\ttB, \vdash\}\to Q\times\{0,1,\ttB, \vdash\}\times \{\leftarrow,\rightarrow\}\times\{0,1,\emptyset\}$ is a partial function. In particular, the tapes of $M$ are infinite to the right so the symbol $\vdash$ is used to indicate the first position in each tape, and $M$ does not have any final states, as it produces an output for each input. Moreover, the only allowed operations in the output tape are: 
 (1) writing 0 and moving the head one cell to the right, (2) writing 1 and moving the head one cell to the right, or (3) doing nothing. These operations are represented by 0, 1, and $\emptyset$, respectively. Finally, assume that $M$, on input $\enc(\A)$ with domain $A = \{1,\dots,n\}$, executes exactly $n^k$ steps on large inputs for a fixed $k \geq 1$. We ignore small inputs since they can be handled separately.

We construct a formula $\alpha$ in an extension of the grammar of $\rqfo(\fo)$ such that $\sem{\alpha}(\A) = f(\enc(\A))$ for each $\R$-structure $\A$. This extension allows defining the operator ${\bf lsfp}$ for multiple functions, analogously to the notion of simultaneous LFP~\cite{L04}.
Let $\bar{x} = (x_1,\ldots,x_k)$ and $\bar{t} = (t_1,\ldots,t_k)$. Then $\alpha$ is defined as:
\begin{align*}
\alpha \ = \ \sa{\bar{t}}\clfp{\out(\bar{t}): \,&\alpha_{T_0}(\bar{t},\bar{x},T_0,T_1,T_{\ttB},T_{\vdash},h,\hat{h},s_{q_0},\ldots,s_{q_{\ell}},\out),\\
	&\alpha_{T_1}(\bar{t},\bar{x},T_0,T_1,T_{\ttB},T_{\vdash},h,\hat{h},s_{q_0},\ldots,s_{q_{\ell}},\out),\\
	&\alpha_{T_{\ttB}}(\bar{t},\bar{x},T_0,T_1,T_{\ttB},T_{\vdash},h,\hat{h},s_{q_0},\ldots,s_{q_{\ell}},\out),\\
	&\alpha_{T_{\vdash}}(\bar{t},\bar{x},T_0,T_1,T_{\ttB},T_{\vdash},h,\hat{h},s_{q_0},\ldots,s_{q_{\ell}},\out),\\
	&\alpha_{h}(\bar{t},\bar{x},T_0,T_1,T_{\ttB},T_{\vdash},h,\hat{h},s_{q_0},\ldots,s_{q_{\ell}},\out),\\
	&\alpha_{\hat{h}}(\bar{t},\bar{x},T_0,T_1,T_{\ttB},T_{\vdash},h,\hat{h},s_{q_0},\ldots,s_{q_{\ell}},\out),\\
	&\alpha_{s_{q_0}}(\bar{t},T_0,T_1,T_{\ttB},T_{\vdash},h,\hat{h},s_{q_0},\ldots,s_{q_{\ell}},\out),\\
	&\vdots \\
	&\alpha_{s_{q_{\ell}}}(\bar{t},T_0,T_1,T_{\ttB},T_{\vdash},h,\hat{h},s_{q_0},\ldots,s_{q_{\ell}},\out),\\
	&\alpha_{\out}(\bar{t},T_0,T_1,T_{\ttB},T_{\vdash},h,\hat{h},s_{q_0},\ldots,s_{q_{\ell}},\out)}\mult \last(\bar{t}).
\end{align*}
In this formula, $T_0,T_1,T_{\ttB},T_{\vdash},h,\hat{h}$ are functions symbols of arity $2 \cdot k$, while $s_{q_0},\ldots,s_{q_{\ell}},\out$ are function symbols of arity $k$. For each one of these function symbols $f$, there is a formula $\alpha_f$ defining how the values of $f$ are updated when computing the fixed point. For example, $\alpha_{T_0}$ is used to define the values of function $T_0$. Notice that the values of each function in $\alpha$ depend on the values of the other functions, which is why we talk about a simultaneous computation. Besides, notice that the notation $\clfp{\out(\bar{t}) : \ldots}$ is used to indicate that (1) the free variables of the formula are $\bar t = (t_1, \ldots, t_k)$, and (2) once the least fixed point has been computed, the value of $\clfp{\out(\bar{t}) : \ldots}$ for an assignment $\bar a$ to $\bar t$ is given by the value (in the fixed point) of function $\out$  on $\bar a$. Finally, it is important to notice that $\alpha$ can be transformed into a proper $\rqfo(\fo)$ formula by using the same techniques used to prove that simultaneous LFP has the same expressiveness as LFP~\cite{L04}; in particular, functions symbols $T_0,T_1,T_{\ttB},T_{\vdash},h,\hat{h},s_{q_0},\ldots,s_{q_{\ell}},\out$ are replaced by a single function $f$ of arity $6 \cdot 2 \cdot k + (\ell + 1) \cdot k + k = (\ell+ 14) \cdot k$.

In the formula $\alpha$, function $T_0$ is used to indicate whether the content of a cell of the working tape is 0 at a certain point of time, that is, $T_0(\bar{t},\bar{x}) > 0$ if the cell at position $\bar{x}$ of the working tape contains the symbol 0 at step $\bar{t}$, and $T_0(\bar{t},\bar{x}) = 0$ otherwise. Functions $T_1$, $T_{\ttB}$ and $T_{\vdash}$ are defined analogously. Function $h$ is used to indicate whether the head of the working tape is in a specific position at a certain point of time, that is, $h(\bar{t},\bar{x}) > 0$ if the head of the working tape is at position $\bar{x}$ at step $\bar{t}$, and $h(\bar{t},\bar{x}) = 0$ otherwise. 
Function $\hat{h}$ is used to indicate whether the head of the working tape is {\bf not} in a position at a particular point of time, that is, $\hat{h}(\bar{t},\bar{x}) > 0$ if the head of the working tape is {\bf not} at position $\bar{x}$ at step $\bar{t}$, and $\hat{h}(\bar{t},\bar{x}) = 0$ otherwise. For each $i \in \{0, \ldots, \ell\}$, function $s_{q_i}$ is used to indicate whether the TM $M$ is in state $q_i$ at a certain point of time, that is, $s_{q_i}(\bar{t}) > 0$ if the TM $M$ is in state $q_i$ at step $\bar{t}$, and $s_{q_i}(\bar{t}) = 0$ otherwise. Function $\out$ stores the output of the TM $M$; in particular, $\out(\bar t)$ is the value returned by $M$ when $\bar t$ is the last step. Finally, assuming that $\varphi_{\text{\rm last}}(x)$ is an $\fo$-formula defining the last element of $<$, we have that $\last(\bar t) = \bigwedge_{i=1}^k \varphi_{\text{\rm last}}(t_i)$, that is, $\last(\bar t)$ holds if $\bar t$ is the last step. Therefore, the use of $\last(\bar t)$ in $\alpha$ allows us to return the output of the TM $M$. 

As in Example \ref{ex:count-path}, assume that $\varphi_{\text{\rm first}}(x)$ is an $\fo$-formula defining the first element of $<$. Moreover, assume that $\first(\bar t)$ and $\succesor(\bar{t},\bar{t'})$ are $\fo$-formulae such that 
$\first(\bar t)$ holds if $\bar t$ is the first step and $\succesor(\bar{t},\bar{t'})$ holds if $\bar t'$ is the successor step of $\bar t$ (that is, $\succesor(\cdot,\cdot)$ is the successor relation associated to the lexicographical order induced by $<$ on the tuples with $k$ elements).  Then formula $\alpha_{T_{\vdash}}$ is defined as follows:
\begin{multline*}
\alpha_{T_{\vdash}}(\bar{t},\bar{x},T_0,T_1,T_{\ttB},T_{\vdash},h,\hat{h},s_{q_0},\ldots,s_{q_{\ell}},\out) = \\
(\first(\bar{t}) \wedge \first(\bar{x})) +
\sa{\bar{t}'}(\successor(\bar{t}',\bar{t}) \mult \hat{h}(\bar{t}',\bar{x}) \mult T_{\vdash}(\bar{t}',\bar{x})) \ + \\
\sum_{(q,a) \,:\, \delta(q,a) = (\_,\vdash,\_,\_)}\sa{\bar{t}'}(\successor(\bar{t}',\bar{t}) \mult h(\bar{t}',\bar{x}) \mult T_a(\bar{t}',\bar{x}) \mult s_{q}(\bar{t}')).
\end{multline*}
Notation $\sum_{(q,a) \,:\, \delta(q,a) = (\_,\vdash,\_,\_)}$ in this formula is used to indicate that we are considering a sum over all pairs $(q,a) \in Q \times \{0,1,\ttB,\vdash\}$ such that $\delta(q,a) = (q',\vdash,u,v)$ for some $q' \in Q$, $u \in \{\leftarrow,\rightarrow\}$ and $v \in \{0,1,\emptyset\}$.
Formulae $\alpha_{T_0}$, $\alpha_{T_1}$ and $\alpha_{T_{\ttB}}$ are defined analogously; for the sake of presentation, we only show here how $\alpha_{T_0}$ is defined, assuming that $\bar y = (y_1, \ldots, y_k)$:
\begin{multline*}
\alpha_{T_0}(\bar{t},\bar{x},T_0,T_1,T_{\ttB},T_{\vdash},h,\hat{h},s_{q_0},\ldots,s_{q_{\ell}},\out) = \\ 
(\first(\bar{t}) \wedge \exists\bar{y}(\varphi_{\text{\rm first}}(y_1) \wedge \cdots \wedge \varphi_{\text{\rm first}}(y_{k-2}) \wedge\neg E(y_{k-1},y_k) \wedge \successor(\bar{y},\bar{x}) ))\ + \\
\sa{\bar{t}'}(\successor(\bar{t}',\bar{t}) \mult \hat{h}(\bar{t}',\bar{x}) \mult T_0(\bar{t}',\bar{x})) + \\
\sum_{(q,a) \,:\, \delta(q,a) = (\_,0,\_,\_)}\sa{\bar{t}'}(\successor(\bar{t}',\bar{t}) \mult h(\bar{t}',\bar{x}) \mult T_a(\bar{t}',\bar{x}) \mult s_{q}(\bar{t}')).
\end{multline*}
Notice that in the initial step, the encoding of the structure $\A$ has to be placed on the working tape of $M$ from the second position since the symbol $\vdash$ is placed in the first position. This is the reason why we use tuple $\bar y$ and define $\bar x$ as the successor of $\bar y$.
Formula $\alpha_{h}$ is defined as follows:
\begin{align*}
\alpha_{h}(\bar{t},\bar{x},&T_0,T_1,T_{\ttB},T_{\vdash},h,\hat{h},s_{q_0},\ldots,s_{q_{\ell}},\out) = \\
&(\first(\bar{t}) \wedge \successor(\bar{t},\bar{x})) + \\
&\sum_{(q,a) \,:\, \delta(q,a) = (\_,\_,\leftarrow,\_)}\sa{\bar{t}'}\sa{\bar{x}'}(\successor(\bar{t}',\bar{t}) \mult \successor(\bar{x},\bar{x}')\mult h(\bar{t}',\bar{x}') \mult T_a(\bar{t}',\bar{x}') \mult s_{q}(\bar{t}')) \ + \\
&\sum_{(q,a) \,:\, \delta(q,a) = (\_,\_,\rightarrow,\_)}\sa{\bar{t}'}\sa{\bar{x}'}(\successor(\bar{t}',\bar{t}) \mult \successor(\bar{x}',\bar{x})\mult h(\bar{t}',\bar{x}') \mult T_a(\bar{t}',\bar{x}') \mult s_{q}(\bar{t}')).
\end{align*}
Similarly, formula $\alpha_{\hat{h}}$ is defined as follows:
\begin{align*}
\alpha_{\hat{h}}(\bar{t},\bar{x},&T_0,T_1,T_{\ttB},T_{\vdash},h,\hat{h},s_{q_0},\ldots,s_{q_{\ell}},\out) = \\
&(\first(\bar{t}) \wedge \neg\successor(\bar{t},\bar{x})) + \\
&\sum_{(q,a) \,:\, \delta(q,a) = (\_,\_,\leftarrow,\_)}\sa{\bar{t}'}\sa{\bar{x}'}\sa{\bar{x}''}(\successor(\bar{t}',\bar{t}) \mult h(\bar{t}',\bar{x}') \mult T_a(\bar{t}',\bar{x}') \mult s_{q}(\bar{t}') \, \mult \\
&\hspace{250pt} \successor(\bar{x}'',\bar{x}') \mult (\bar{x} \neq \bar{x}'')) \ + \\
&\sum_{(q,a) \,:\, \delta(q,a) = (\_,\_,\rightarrow,\_)}\sa{\bar{t}'}\sa{\bar{x}'}\sa{\bar{x}''}(\successor(\bar{t}',\bar{t}) \mult h(\bar{t}',\bar{x}') \mult T_a(\bar{t}',\bar{x}') \mult s_{q}(\bar{t}') \, \mult \\
&\hspace{250pt} \successor(\bar{x}',\bar{x}'') \mult (\bar{x} \neq \bar{x}'')).
\end{align*}
Notice that in the definition of $\alpha_{\hat{h}}$, $\fo$-formula $\bar{x} \neq \bar{x}''$ is used to indicate that tuples $\bar{x}$ and $\bar{x}''$ are distinct (that is, there exists $i \in \{1, \ldots, k\}$ such that the $i$-th components of $\bar{x}$ and $\bar{x}''$ are distinct). 
Formula $\alpha_{s_{q_0}}$ is defined as:
\begin{multline*}
\alpha_{s_{q_0}}(\bar{t},T_0,T_1,T_{\ttB},T_{\vdash},h,\hat{h},s_{q_0},\ldots,s_{q_{\ell}},\out) \ = \\ 
\first(\bar{t}) \ + \
\sum_{(q,a) \,:\, \delta(q,a) = (q_0,\_,\_,\_)}\sa{\bar{t}'}\sa{\bar{x}'}(\successor(\bar{t}',\bar{t}) \mult h(\bar{t}',\bar{x}') \mult T_a(\bar{t}',\bar{x}') \mult  s_{q}(\bar{t}')).
\end{multline*}
Moreover, for every $i \in \{1, \ldots, \ell\}$, formula $\alpha_{s_{q_i}}$ is defined analogously.
Finally, formula $\alpha_{\out}$ is defined as:
\begin{align*}
\alpha_{\out}(&\bar{t},T_0,T_1,T_{\ttB},T_{\vdash},h,\hat{h},s_{q_0},\ldots,s_{q_{\ell}},\out) =\\
&\sum_{(q,a) \,:\, \delta(q,a) = (\_,\_,\_,0)}\sa{\bar{t}'}\sa{\bar{x}'}(\successor(\bar{t}',\bar{t}) \mult h(\bar{t}',\bar{x}') \mult T_a(\bar{t}',\bar{x}') \mult s_{q}(\bar{t}') \mult 2\mult \out(\bar{t}')) \ + \\
&\sum_{(q,a) \,:\,\delta(q,a) = (\_,\_,\_,1)}\sa{\bar{t}'}\sa{\bar{x}'}(\successor(\bar{t}',\bar{t}) \mult h(\bar{t}',\bar{x}') \mult T_a(\bar{t}',\bar{x}') \mult s_{q}(\bar{t}') \mult (2\mult \out(\bar{t}')+1)) \ + \\ 
&\sum_{(q,a) \,:\,\delta(q,a) = (\_,\_,\_,\emptyset)}\sa{\bar{t}'}\sa{\bar{x}'}(\successor(\bar{t}',\bar{t}) \mult h(\bar{t}',\bar{x}') \mult T_a(\bar{t}',\bar{x}') \mult s_{q}(\bar{t}') \mult \out(\bar{t}')).
\end{align*}
Clearly, tuple $\bar{t}$ encodes the number of steps the machine has done and in each iteration of the fixed point operator, one timestep of the machine is executed.
Assume that $\bar{a}$, $\bar{a}'$ are tuples with $k$ elements such that $\bar{a}'$ is the successor of $\bar{a}$. 
From the construction of the formula, and since the machine is deterministic, it can be seen that for each function $f \in\{T_0,T_1,T_{\ttB},T_{\vdash},h,\hat{h}\}$, at the $\bar{a}$-th iteration of the fixed point operator, it holds that $f(\bar{a},\bar{b}) \leq 1$ and  $f(\bar{a}',\bar{b}) = 0$ for every $\bar{b}\in A^k$. In the same way, it can be seen that (1) for each function $g \in\{s_{q_1},\ldots,s_{q_{\ell}}\}$, at the $\bar{a}$-th iteration of the fixed point operator, it holds that $g(\bar{a}) \leq 1$ and $g(\bar{a}') = 0$, that (2) for each $\bar{b}$, only one of the functions $T_0$,$T_1$,$T_B$,$T_{\vdash}$ outputs 1 on input ($\bar{a}, \bar{b})$, (3) for each $\bar{b}$, only one of the functions $h,\hat{h}$ outputs 1 on input $(\bar{a}, \bar{b})$, (4) there is exactly one $\bar{b}$ with $h(\bar{a}, \bar{b}) = 1$, and (5) there is exactly one $i$ such that $s_{q_i}(\bar{a}) = 1$. From this, we have that at the iteration $\bar{a}'$ of the fixed point operator, $\out(\bar{a}')$ is equal to either $2\cdot \out(\bar{a})$, $2\cdot \out(\bar{a}) + 1$, or $\out(\bar{a})$, which represents precisely the value in the output tape at each step of $M$ running on the input $\enc(\A)$. From this argument, it can be seen that $\sem{\alpha}(\A) = f(\enc(\A))$, which concludes the proof of the theorem.

\end{proof}
Our last goal in this section is to use the new characterization of $\fp$ to explore classes below it.
It was shown in \cite{I86,I88} that $\fo$ extended with a transitive closure operator captures $\nlog$. 
Inspired by this work, we show that a restricted version of $\rqfo$ can be used to capture $\shl$, the counting version of $\nlog$. 
Specifically, we use $\rqfo$ to define an operator for counting the number of paths in a directed graph, which is what is needed to capture~$\shl$.

Given a relational signature $\R$, the set of transitive $\qfo$ formulae ($\tqfo$-formulae) is defined as the extension of $\qfo$ by the additional operator $[\pth \psi(\bar{x},\bar{y})]$, where $\psi(\x, \y)$ is an $\fo$-formula over $\R$, and $\bar{x} = (x_1, \ldots, x_k)$, $\bar{y} = (y_1, \ldots, y_k)$ are tuples of pairwise distinct first-order variables. The semantics of $[\pth \psi(\bar{x},\bar{y})]$ can easily be defined in terms of $\rqfo(\fo)$ as follows. 
Given an $\R$-structure $\A$ with domain $A$, define a (directed) graph $\cG_{\psi}(\A) = (N,E)$ such that $N = A^k$ and for every pair $\bar b, \bar c \in N$, it holds that $(\bar b, \bar c) \in E$ if, and only if, $\A \models \psi(\bar b, \bar c)$.
As was the case for Example~\ref{ex:count-path}, we can count the paths of length at most $|A^k|$ in $ \cG_{\psi}(\A)$ with the formula $\beta_{\psi(\bar{x},\bar{y})}(\x, \y, \t, g)$:
\begin{align*}
(\psi(\bar{x},\bar{y}) + \sa{\z} g(\x,\z,\t)\cdot \psi(\z,\y)) \cdot \varphi_{\text{\rm first-lex}}(\t) \ +
\sa{\t'} \varphi_{\text{succ-lex}}(\t',\t) \cdot \left(\sa{\x'} \sa{\y'} g(\x',\y',\t') \right),
\end{align*}
where $\varphi_{\text{\rm first-lex}}$ and $\varphi_{\text{succ-lex}}$ are $\fo$-formulae defining the first and successor predicates over tuples in $A^k$, following the lexicographic order induced by~$<$.
Then the semantics of the path operator can be defined by using the following definition of $[\pth \psi(\bar{x},\bar{y})]$ in $\rqfo$:
\begin{eqnarray*}
[\pth \psi(\bar{x}, \bar{y})] & := & \sa{\t} (\varphi_{\text{\rm first}}(\t) \cdot \clfp{\beta_{\psi(\bar{x},\bar{y})}(\x,\y,\t,g)}).
\end{eqnarray*}
In other words, $\sem{[\pth \psi(\bar{x}, \bar{y})]}(\A,v)$ counts the number of paths from $v(\bar x)$ to $v(\bar y)$ in the graph $\cG_{\psi}(\A)$ whose length is at most~$|A^k|$.
As mentioned before, the operator for counting paths is exactly what we need to capture $\shl$.
\begin{thm} \label{tqfo-shl}
	$\tqfo(\fo)$ captures $\shl$ over ordered structures.
\end{thm}
\begin{proof}

First, we show that every formula in $\tqfo(\fo)$ defines a function that is in $\shl$.
Let $\R$ be a relational signature and $\alpha$ a formula over $\R$ in $\tqfo(\fo)$. 
We construct a logarithmic-space nondeterministic Turing Machine $M_{\alpha}$ that on input $(\enc(\A),v)$, where $\A$ is an $\R$-structure and $v$ is a first-order assignment for $\A$, has $\sem{\alpha}(\A,v)$ accepting runs (so that we can conclude that the function defined by $\alpha$ is in $\shl$). 
Suppose that the domain of $\A$ is $A = \{1,\ldots,n\}$. The TM $M_{\alpha}$ needs $\ell \mult\log_2(n)$ bits of memory to store the first-order variables that appear in $\alpha$, where $\ell$ is the number of different variables in this formula (and also the size of the domain of $v$). 
If $\alpha = \varphi$, where $\varphi$ is an $\fo$-formula, then we check if $(\A,v)\models\varphi$ in deterministic logarithmic space, and accept if and only if this condition holds. If $\alpha = s$, where $s$ is a fixed natural number, then we generate $s$ branches and accept in all of them. 
If $\alpha = (\alpha_1 + \alpha_2)$, we simulate $M_{\alpha_1}$ and $M_{\alpha_2}$ on separate branches. If $\alpha = (\alpha_1\mult\alpha_2)$, we simulate $M_{\alpha_1}$ and if it accepts, then instead of accepting we simulate $M_{\alpha_2}$. If $\alpha = \sa{x}\beta$, for each $a\in A$ we generate a different branch where we simulate $M_{\beta}$ with input $v[a/x]$. If $\alpha = \pa{x}\beta$, we simulate $M_{\beta}$ with input $v[1/x]$, and on each accepting run, instead of accepting we simulate $M_{\beta}$ with input $v[2/x]$, and so on.
If $\alpha = [\pth \varphi(\bar{x},\bar{y})]$, where $\varphi$ is an $\fo$-formula, we simulate the $\shl$ procedure that counts the number of paths of a given length from a source to a target node in an input graph (where the length is at most the number of nodes in the graph).

Second, we show that every function in $\shl$ can be encoded by a formula in $\tqfo(\fo)$.
Let $f$ be a function in $\shl$ and $M$ a logarithmic-space nondeterministic  Turing Machine such that $\tma_M(\enc(\A)) = f(\enc(\A))$. We assume that $M$ has only one accepting state, and that no transition is defined for this state. Moreover, we assume that there exists only one accepting configuration. We make use of transitive closure logic ($\tc$) to simplify our proof~\cite{G07}. We have that $\tc$ captures $\nlog$\cite{I83}, so that there exists a formula $\varphi$ in $\tc$ such that $\A\models\varphi$ if and only if $M$ accepts $\enc(\A)$. This formula can be expressed as:
$$
\varphi = \exists\bar{u}\exists\bar{z}(\psi_{\text{initial}}(\bar{u})\wedge \psi_{\text{acc}}(\bar{z})\wedge[{\bf tc}_{\bar{x},\bar{y}}\,\psi_{\text{next}}(\bar{x},\bar{y})](\bar{u},\bar{z})),
$$
where $\psi_{\text{initial}}(\bar{u})$ is an $\fo$-formula that indicates that $\bar{u}$ is the initial configuration, $\psi_{\text{acc}}(\bar{z})$ is an $\fo$-formula that indicates that $\bar{z}$ is an accepting configuration, and $\psi_{\text{next}}(\bar{x},\bar{y})$ is an $\fo$-formula that indicates that $\bar{y}$ is a successor configuration of $\bar{x}$~\cite{G07}. Here, there is a one-to-one correspondence between configurations of $M$ and assignments to $\bar{z}$. As a consequence, given a structure $\A$ and a first-order assignment $v$ for $\A$, where $v(\bar{x})$ is the starting configuration and $v(\bar{y})$ is the sole accepting configuration, the value of $\sem{[\pth\psi_{\text{next}}(\bar{x},\bar{y})]}(\A,v)$ is equal to $\tma_M(\enc(\A))$.
Therefore, the $\tqfo(\fo)$-formula
$
\alpha = \sa{\bar{u}}\sa{\bar{z}}(\psi_{\text{initial}}(\bar{u})\mult\psi_{\text{acc}}(\bar{z})\mult[\pth \psi_{\text{next}}(\bar{u},\bar{z})])
$
satisfies that $\sem{\alpha}(\A) = f(\enc(\A))$. This concludes the proof of the theorem.
%

\end{proof}
This last result 
perfectly 
illustrates the benefits of our logical framework for the development of descriptive complexity for counting 
complexity classes.  
The distinction in the language between the Boolean and the quantitative level allows us to define operators at the latter level that cannot be defined at the former. 
As an example showing how fundamental this separation is, consider the issue of extending $\qfo(\fo)$ at the Boolean level in order to capture $\shl$. The natural alternative to do this is to use $\fo$ extended with a transitive closure operator, which is denoted by $\tc$. But then the problem is that for every language $L \in \nlog$, it holds that its characteristic function $\chi_L$ is in $\qfo(\tc)$, where $\chi_L(x) = 1$ if $x \in L$, and $\chi_L(x) = 0$ otherwise. Thus, if we assume that $\qfo(\tc)$ captures $\shl$ (over ordered structures), then we have that $\chi_L \in \shl$ for every $L \in \nlog$. This would imply that $\nlog = \ulog$,\footnote{A decision language $L$ is in $\ulog$ is there exists a logarithmic-space NTM $M$ accepting $L$ and satisfying that $\tma_M(x) = 1$ for every $x \in L$.} solving an outstanding open problem \cite{Reinhardt97}.


\section{Concluding remarks and future work}\label{sec:conclusions}

We proposed a framework based on Weighted Logics to develop a descriptive complexity theory for complexity classes of functions.
We consider the results of this paper as a first step in this direction.
Consequently, there are several directions for future research, some of which are mentioned here. 
$\totp$ is an interesting counting complexity class as it naturally defines a class of functions in $\shp$ with easy decision counterparts. However, we do not have a logical characterization of this class.
Similarly, we are missing logical characterizations of other fundamental complexity classes such as $\spanl$~\cite{AlvarezJ93}. We would also like to define a larger syntactic subclass of $\shp$ where each function admits an FPRAS; notice that $\cpm$ is an important problem admitting an FPRAS\cite{JSV04} that is not included in the classes defined in Section \ref{sec-clo}. Moreover, by following the approach proposed in
\cite{I83}, we would like to include second-order free variables in the operator for counting paths introduced in Section \ref{sec:beyond}, in order to have alternative ways to capture $\fpspace$ and even $\shp$. 

Finally, an open problem is to understand the connection of this framework with enumeration complexity classes. For example, in~\cite{durandS11} the enumeration of assignments for first-order formulas was studied following the $\sfo$-hierarchy, and in~\cite{AmarilliBJM17} the efficient enumeration (i.e.\ constant delay enumeration~\cite{Segoufin13}) was shown for a particular class of circuits. It would be interesting then to adapt the $\eqso(\fo)$-hierarchy for the context of enumeration and identify subclasses that also admit good properties in terms of enumeration.


\section{Acknowledgements}
The authors are grateful to Luis Alberto Croquevielle for providing the proof of Proposition \ref{prop:hsat-not-sigma2}.
This research has been supported by the Fondecyt grant 1161473 and the Millennium Institute for Foundational Research on Data.

\bibliographystyle{alpha}
\bibliography{biblio}

\end{document}